\documentclass[journal]{IEEEtran}

\usepackage{amsmath,epsfig,amssymb,verbatim,amsopn,cite,subfigure,multirow}
\usepackage{balance}

\usepackage{gensymb}
\allowdisplaybreaks

\usepackage[usenames,dvipsnames]{color}
\usepackage[all]{xy}  
\usepackage{url}
\usepackage{amsfonts}
\usepackage{amssymb}
\usepackage{epsfig}
\usepackage{amsmath,epsfig,amssymb,verbatim,amsopn,cite,subfigure,multirow}
\usepackage{mathrsfs}
\usepackage{amsmath,epsfig,amssymb,verbatim,amsopn,cite,subfigure,multirow}
\usepackage{balance}
\usepackage[usenames,dvipsnames]{color}
\usepackage[all]{xy}  
\usepackage{url}
\usepackage{amsfonts}
\usepackage{amssymb}
\usepackage{epsfig}
\usepackage{cite}
\usepackage{graphicx}
\graphicspath{ {image/} }
\usepackage{epstopdf}
\usepackage{algorithm}
\usepackage{algpseudocode}
\usepackage{enumitem}

\usepackage{import, times, graphicx,epsfig,fancybox,amsfonts, comment, float, setspace, color, url, amsthm, amsmath}
\newcommand{\Gtot}{G_{tot}}

\newcommand{\dist}{\operatorname{dist}}

\DeclareMathOperator*{\argmin}{arg\,min}

\newcommand*\diff{\mathop{}\!\mathrm{d}}

\newtheorem{theorem}{Theorem}
\newtheorem{lemma}[theorem]{Lemma}
\newtheorem{cor}[theorem]{Corollary}

\newtheorem{prop}[theorem]{Proposition}

\newtheorem{definition}{Definition}
\newtheorem{example}{Example}
\newtheorem{remark}{Remark}

\title{The Gain of Energy Accumulation in Multi-hop Wireless Network Broadcast}

\author{\IEEEauthorblockN{Majid~Khabbazian, and Keyvan~Gharouni~Saffar\\
\IEEEauthorblockA{Department of Electrical and Computer Engineering\\
		University of Alberta, Canada\\
		\{mkhabbazian, gharouni\}@ualberta.ca}
}}

\begin{document}
\maketitle

\begin{abstract}
Broadcast is a fundamental network operation, widely used in wireless networks to disseminate messages.
The energy-efficiency of broadcast is important particularly when devices in the network are energy constrained.
To improve the efficiency of broadcast, different approaches have been taken in the literature. 
One of these approaches is broadcast with energy accumulation.
Through simulations, it has been shown in the literature that broadcast with energy accumulation can result in energy saving.
The amount of this saving, however, has only been analyzed for linear multi-hop wireless networks.
In this work, we extend this analysis to two-dimensional (2D) multi-hop networks.
The analysis of saving in 2D networks is much more challenging than that in linear networks.
It is because, unlike in linear networks, in 2D networks, finding minimum-energy broadcasts with or without energy accumulation are both NP-hard problems.
Nevertheless, using a novel approach, we prove that this saving is constant when the path loss
exponent $\alpha$ is strictly greater than two.
Also, we prove that the saving is $\theta(\log n)$ when $\alpha=2$, 
where $n$ denotes the number of nodes in the network.



\end{abstract}
%
\section{Introduction}
Applications of wireless multi-hop networks  ranges from Smart Grid \cite{blumsack2012ready,mcdaniel2009security} and Internet of Things (IoT) \cite{mattern2010internet,atzori2010internet} to Machine-to-Machine (M2M) communications networks \cite{niyato2011machine,wu2011m2m} and smart environments \cite{cook2004smart,poslad2011ubiquitous}. 
Network-wide broadcast, simply referred to as broadcast, is a fundamental operation frequently used in these applications.
Reducing the energy consumption of broadcast is important, because many wireless devices (e.g., mobile devices) have limited energy supplies as they run on batteries.

Different techniques have been used in the literature to reduce the energy consumption of broadcast.
These techniques are typically based on transmission power control \cite{aziz2013survey}, reducing the number of redundant transmissions\cite{khabbazian2012local,nikolov2015towards}, and cooperative communication.
The cooperative communication technique includes three main approaches: 
i) coherent signal synchronization~\cite{KAMZ2003}; 
ii) mutual information accumulation~\cite{MMYZ2007}; and 
iii) energy accumulation~\cite{maric2004cooperative}. 

The coherent signal synchronization approach requires transmitters to synchronize their transmission at the signal level when transmitting
to a single receiver.
This approach offers higher benefits than information accumulation and energy accumulation approaches, but is hard to implement in practice 
because of its tight synchronization requirement~\cite{baghaie2011delay}.
In addition, the benefit of this approach is not clear in a multi-hop network, where there are multiple receivers~\cite{baghaie2011delay}.

The second approach is mutual-information accumulation, which can be implemented using fountain codes~\cite{MMYZ2007}.
In this approach, a node stores partial information from multiple received transmissions, 
and can decode the message when the total accumulated information from previous transmissions of the message exceeds the entropy of the message. 
At low signal-to-noise ratios (SNRs), this approach was shown to be equivalent to the energy accumulation approach~\cite{DLMY08ICC}.

In energy accumulation, which is the subject of this work, a node can decode a message if the total accumulated energy 
from the previous transmissions of the message is above a threshold \cite{qiu2014energy}.
This generalizes the conventional non-cooperative broadcast models in which a node can decode the message 
only if the received power from a single transmission of the message is above a threshold.
For example, suppose the decoding threshold is 1.0, and a node $u$ receives exactly three transmitted signals with received powers $0.5$, $0.4$,
and $0.3$.
With energy accumulation, node $u$ is able to decode since the sum of the received powers is $1.2$, which is above the threshold.
In non-cooperative broadcast, however, decoding is not possible because, for each transmission, the received power at $u$ is strictly less than the threshold.


Energy accumulation can let transmitters reduce their transmission powers while still assuring other nodes will receive the message.
This, as shown in \cite{maric2004cooperative} and \cite{baghaie2011delay} through simulations, 
can reduce the power consumption of broadcast, hence saving energy.
Theoretical analysis of the amount of this energy saving, however, can shed more light on how much gain/saving can be achieved using energy accumulation.

If the minimum energy consumption of cooperative broadcast and non-cooperative broadcasts are $P_c$ and $P_n$, respectively, 
then the cooperation gain (i.e., the energy saving) is defined as $\frac{P_n}{P_c}$.
To evaluate the cooperation gain, therefore, we need to find values of $P_n$ and $P_c$.
Unfortunately, computing $P_n$ and $P_c$ are both NP-hard~\cite{li2001minimum,vcagalj2002minimum}.
Nevertheless, we evaluate the cooperation gain $\frac{P_n}{P_c}$ using a novel approach.
 
In our approach, we first design a conversion method that converts any cooperative broadcast algorithm into a non-cooperative broadcast algorithm.
Then, we  prove that the ratio of the energy consumption of the constructed non-cooperative algorithm over that of the cooperative algorithm 
is constant when the path loss exponent is strictly greater than two.
Therefore, if there is a cooperative broadcast algorithm with power consumption $P_c$, then there exists a non-cooperative broadcast with
power consumption $\mathcal{O}(1)\cdot P_c$, thus the cooperation gain is $\mathcal{O}(1)$.
Using the same approach we show that the cooperation gain is $\mathcal{O}(\log n)$ when $\alpha=2$,
where $n$ denotes the number of nodes in the network.

The main  contributions of this work are: 
\begin{enumerate}
	\item 
	  When the path loss exponent is greater than two (the typical case in practice), we show that the gain of energy accumulation in 2D networks is limited (constant)
	  irrespective of the number of nodes, and their locations in the network.
	\item
	  When $\alpha=2$, we prove that the gain grows at best logarithmically with the number of nodes in the network.  
    \item
      When $\alpha=2$, we show that there are networks such as grid and random networks, in which the gain grows logarithmically with the number
      of nodes in the network.
\end{enumerate}

This paper is organized as follows.
We discuss related work in Section~\ref{related work}.
System model and definitions are described in Section~\ref{sec:model}.  Our analytical results including bounds on the cooperation gain are presented in Section~\ref{AnGtot}. In Section~\ref{third:simulation}, we verify some of our analytical results through simulation. Finally, Section~\ref{conclu} concludes the paper.

%
%
%
%

\section{Related Work}
\label{related work}
Energy accumulation can be performed at receivers utilizing maximal ratio combining (MRC) of orthogonal signals in time, frequency or code domain (see \cite{sirkeci2007power,baghaie2011delay,qiu2014energy}). 

Energy accumulation has been studied for both routing~\cite{Chen05,Gomez17}, 
and broadcasting~\cite{Maricand2002, Agarwal2004} in wireless multi-hop networks.
It has also been used to enhance throughput in cognitive radio networks~\cite{Sharifi12, Zhang13, Atat18}. 
The scope of our work is to analyze the energy saving that can be achieved by using energy accumulation for broadcast in wireless multi-hop networks.

Existing energy accumulation based cooperative broadcast algorithms fall into two groups.
The first group includes algorithms (e.g., \cite{hong2006energy,sirkeci2007power}) in which receiving nodes can combine signals from all previous transmissions to benefit from transmission diversity.
These algorithms are called cooperative broadcast algorithms with memory. 
The other group includes ``memoryless'' cooperative broadcast algorithms such as the one proposed in \cite{baghaie2011delay}.
In these algorithms, a node can only use transmissions in the present time slot to accumulate energy; Signals received from transmissions in previous time slots are discarded.
Our work studies cooperative broadcast algorithms with memory as they fully benefit from the energy accumulation. 
As the result, our derived upper bounds on the cooperation gain also apply to the  ``memoryless'' cooperative broadcast algorithms.

The problem of cooperative broadcast with minimum energy can be broken into two sub-problems i) transmission scheduling, which determines the set of transmitters and the order of transmissions;
ii) power allocation, in which the transmission powers are set.
It was proven that, given a transmission scheduling, optimal power allocation can be computed in polynomial time, 
but finding an optimal scheduling that leads to a minimum power consumption is NP-hard \cite{hong2006energy,maric2004cooperative}.

In addition to saving energy, energy accumulation can be used to reduce broadcast latency~\cite{Mergen06}.
Some existing work study the tradeoff between energy and latency in cooperative broadcast~\cite{baghaie2011delay, qiu2014energy}.
In~\cite{baghaie2011delay}, Baghaie and Krishnamachari prove that the problem of
minimizing the energy consumption while meeting a desired latency constraint is not only NP-hard 
but also $o(\log n)$ inapproximable.

The problem of finding non-cooperative broadcast with minimum energy was also proven to be NP-hard \cite{vcagalj2002minimum}. 

Several approximation algorithms have been proposed in the literature for the case where nodes lie in a Euclidean space (e.g.,~\cite{WieselthierNE00, vcagalj2002minimum, CartignySS03}). 
Two well-known approximation algorithms are MST heuristic and BIP (broadcast incremental power)~\cite{WieselthierNE00}.
In 2D networks, these algorithms were proven to have approximation ratios of 6 and 4.33, respectively~\cite{Ambuhl05, WanCLF01}.

The best existing approximation algorithm to the problem is, however, due to Caragiannis, Flammini, and Moscardelli \cite{caragiannis2013exponential}. 
In 2D wireless networks, their algorithm has an approximation ratio of 4.2 for Euclidean cost graphs, and a logarithmic approximation for non-Euclidean cost graphs.

Unlike 2D networks, in linear networks, 
the minimum energy of both cooperative and non-cooperative broadcast algorithms can be computed in polynomial time~\cite{vcagalj2002minimum}.
For linear networks, the ratio of the two minimum power consumptions was proven to be constant 
with respect to the number of nodes in the network~\cite{mobini2016asymptotic}.
Our work extends this study to general 2D networks. 
\section{System Model and Definitions}\label{sec:model}
We consider a static 2D wireless network with a set of $n$ nodes $\mathscr{U}$, and a single source node $s \in \mathscr{U}$, 
which is to broadcast a single message to every other node in $\mathscr{U}$ in a multi-hop fashion. 
We adopt the simplified path loss model used in \cite{maric2004cooperative}, \cite{hong2006energy} and \cite{maric2005cooperative}.
This model is commonly used for system design as it captures the essence of signal propagation 
without resorting to complicated statistical models~\cite{Goldsmith2005}.
In the simplified path loss model, the link gain $\textit{h}_{u_i,u_j}$ is represented as $\textit{h}_{u_i,u_j} = \textit{d}_{u_i,u_j}^{-\alpha}$,
where $\textit{d}_{u_i,u_j}$ is the distance between nodes $\textit{u}_i$ and $\textit{u}_j$ and $\alpha$ is the path loss exponent.
 
We assume that $\alpha \geq 2$ as this is normally the case in practice~\cite{Rappaport96}. 



  Similar to~\cite{maric2004cooperative}, we focus on energy saving without latency constraints.
  This is motivated by applications for networks such as wireless sensor networks, 
  where energy-efficiency is the primary goal.

Any broadcast algorithm can be converted into a collision-free broadcast algorithm in which every node transmits at most once.
To do this, if two nodes transmit at the same time, those transmissions can be separated in time to avoid collisions.
Clearly this does not affect the total power consumption (this may impact latency, which is not the concern of this work).
Also, if a node transmits multiple times, the power of those transmissions can be added together and used  in a single transmission.
For example, if a node $u$ transmits, say twice at times $t_1$ and $t_2$ ($t_1 \leq t_2$) with powers $P_1$ and $P_2$, respectively, 
those transmissions can be merged into one transmission at time $t_1$ with power $P_1 + P_2$.
Therefore, for every broadcast algorithm, there exists a collision-free broadcast algorithm with the same total power consumption in which 
every node transmits at most once.
Since we focus on the total power consumption, we can safely assume that 1) each node transmits at most once, 
and 2) transmissions are collision free.

%

Using the above mode, a broadcast algorithm can be formally defined as follows:
\begin{definition}[Broadcast Algorithm]
	\label{def:BA1}
	We represent a broadcast algorithm by a tuple ($\mathscr{A}, <, \rho$), where the binary relation $<$ is a strict total order on $\mathscr{U}$, and the transmit power function $\rho:\mathscr{U} \to \mathbb{R}_{\geq0}$ is a function from set of nodes to real numbers $\mathbb{R}_{\geq0}$ which returns the transmission power of each node. 
	The set $\mathscr{A} \subseteq \mathscr{U}$ represents the set of nodes that transmit (i.e. set of nodes $u$ for which $\rho(u) \not= 0$). 
	The binary relation determines the order of transmission. 
	In particular, we have
	\begin{itemize}
		\item $u_i<u_j$ if $u_i,u_j \in \mathscr{A}$ and $u_i$ transmits before $u_j$. 
		\item $u_i<u_j$ if $u_i \in \mathscr{A}$ and $u_j \in \mathscr{\bar{A}} = \mathscr{U}\backslash \mathscr{A}$.
		\item $u_i=u_j$ if $u_i,u_j \in \mathscr{\bar{A}}$.
	\end{itemize}
	The binary relation $<$ helps in express equations (e.g. Equation~\ref{cond:coop})  concisely.
\end{definition}
\noindent
In the rest of this paper, we represent a broadcast algorithm by $(\mathscr{A},\rho)$ instead of $(\mathscr{A},<,\rho)$, for convenience.


Definition~\ref{def:BA1} of broadcast algorithm does not guarantee that the message is delivered to all the nodes,
as it does not put any condition on the transmission power $\rho$.
In this work, we only study broadcast algorithms that guarantee full delivery.
To achieve full delivery, extra conditions need to be imposed on the transmit power function.
The following two definitions enforce such conditions, and describe two general classes of broadcast algorithms with full delivery.   

\begin{definition}[Cooperative Broadcast Algorithm]
	A broadcast algorithm ($\mathscr{A},\rho$) is called a cooperative broadcast algorithm,
	and is denoted by ($\mathscr{A}^{(c)},\rho^{(c)}$), if and only if
	
	\begin{equation}
	\label{cond:coop}
	\forall u_{i} \in \mathscr{U}\backslash\{s\}: \quad
	\sum_{u_j<u_i}^{}\rho^{(c)}(u_j)\textit{h}_{u_i, u_j} \geq P_{th},
	\end{equation}
	where $P_{th}$ is the decoding threshold.
	The inequality implies that the sum of received powers at every node (except the source) is not less than the threshold,
	hence every node receives the message.
	
\end{definition}

\begin{definition}[Non-cooperative Broadcast Algorithm]
	A broadcast algorithm ($\mathscr{A},\rho$) is called a non-cooperative broadcast algorithm,
	and is denoted by ($\mathscr{A}^{(n)},\rho^{(n)}$), if and only if:
	\begin{equation}
	\label{cond:non-coop}
	\forall u_{i} \in \mathscr{U} \backslash \{s\},
	\exists u_j < u_i:	\quad
	\rho^{(n)}(u_j)\textit{h}_{u_i,u_j} \geq P_{th}.
	\end{equation}
	Similar to (\ref{cond:coop}), Inequality~\ref{cond:non-coop} ensures full delivery.
\end{definition}

\noindent
The cooperation gain $G_{tot}$ is defined as follows:
\begin{definition}[$\Gtot$]
	\label{Gtot}
	Let ($\mathscr{A}^{(c)}_\dag,\rho^{(c)}_\dag$) and  ($\mathscr{A}^{(n)}_\dag,\rho^{(n)}_\dag$) be, respectively, 
	optimal cooperative and non-cooperative broadcast algorithms with minimum power consumption. 
	We define the cooperation gain $G_{tot}$ as 
	\begin{align}\label{eqn:gtot}
	\Gtot= \frac{\sum_{u_i \in \mathscr{U}}^{}\rho^{(n)}_\dag(u_i)}{\sum_{u_i \in \mathscr{U}}^{}\rho^{(c)}_\dag(u_i)}.
	\end{align}
\end{definition}
\vline

\noindent
Without loss of generality, we assume $P_{th}=1$  in the remaining of the paper.

\section{ANALYSIS OF $G_{tot}$}
\label{AnGtot}

  Before delving into the technical details, let us use a simple network to provide some intuition behind our results, 
  e.g., why the gain grows differently for cases  $\alpha=2$, and  $\alpha>2$.

  As shown in Fig. \ref{grd:2D}, consider a \emph{grid network} with $n$ nodes and minimum distance $d$ between nodes.
  Let $(\mathscr{A}^{(c)},\rho^{(c)})$ be any cooperative broadcast algorithm.
  When a node $u$ transmits, the sum of received powers at nodes other than $u$ and the source~$s$ is
  \begin{equation}
  \label{equ:sumPow}
    \sum_{v\notin \{u,s\}}\frac{\rho^{(c)}(u)}{d_{u,v}^\alpha},
  \end{equation}
  where $d_{u,v}$ denotes the distance between nodes $u$ and $v$.
  If $v$ is at coordinate $(i,j)$ with respect to $u$, then we can simply write $d_{u,v}=d\cdot C_{i,j}$, where $C_{i,j}=(i^2+j^2)^{\frac{1}{2}}$ is the distance
  between coordinate $(i,j)$ and the origin in a Cartesian coordinate system.
  Therefore, by rewriting~(\ref{equ:sumPow}), we get
  \[
    \sum_{v\notin \{u,s\}}\frac{\rho^{(c)}(u)}{d_{u,v}^\alpha} 
    =  \frac{\rho^{(c)}(u)}{d^\alpha} \cdot \underbrace{\sum_{i,j}\frac{1}{C_{i,j}^\alpha}}_{\omega_u}
    =\omega_u \cdot  \frac{\rho^{(c)}(u)}{d^\alpha}.
  \]  
  Note that $\sum_{i,j}\frac{1}{C_{i,j}^\alpha}$ takes its maximum value when $u$ is at the center of the grid.
 Therefore, $\omega_u\leq \zeta_\alpha$, where
  \begin{equation}
    \label{equ:zeta}
    \zeta_\alpha = \sum_{-\lfloor \frac{\sqrt{n}-1}{2}\rfloor \leq i,j \leq \lceil \frac{\sqrt{n}-1}{2}\rceil; (i,j)\neq (0,0)} \frac{1}{\left(i^2+j^2\right)^\frac{\alpha}{2}}
  \end{equation}
  
  Now, notice that in any cooperative broadcast algorithm, the total received powers at every node (except the source node) must be at least one.
  Therefore, since we have $n-1$ nodes other than the source, we must have
  \[
    \sum_{u\in \mathscr{U}} \omega_u \cdot  \frac{\rho^{(c)}(u)}{d^\alpha} \geq n-1,
  \]
  where $\mathscr{U}$ is the set of nodes.
  The above inequality is equivalent to
  \[
    \sum_{u\in \mathscr{U}} \omega_u \cdot  \rho^{(c)}(u)\geq (n-1)d^\alpha.
  \]
  Since  $\omega_u\leq \zeta_\alpha$ for every node $u\in\mathscr{U}$, we get
  \begin{equation}
  \label{equ:omg}
    \sum_{u\in \mathscr{U}}\rho^{(c)}(u)\geq \frac{1}{\zeta_\alpha}\cdot(n-1)d^\alpha.
  \end{equation}
  In other words, the total power consumption of any cooperative broadcast in the grid network is at least $\frac{1}{\zeta_\alpha}\cdot(n-1)d^\alpha$.
  
  On the other hand, non-cooperative broadcast can deliver the message with total power consumption of at most $(n-1)d^\alpha$.
  It is because the algorithm in which every node in the gird (except the last one) transmits with power $d^\alpha$ is a non-cooperative broadcast algorithm,
  and has total power consumption of $(n-1) d^\alpha$.
  Consequently, by~(\ref{equ:omg}), we get the following bound on the cooperation gain
  \begin{equation}
  \label{equ:zetGtot}
    G_{tot}\leq \frac{(n-1) d^\alpha}{\frac{1}{\zeta_\alpha}\cdot(n-1)d^\alpha}=\zeta_\alpha.
  \end{equation}

  It is not hard to show that $\zeta_\alpha$ (defined in~(\ref{equ:zeta})) is constant when $\alpha>2$, and is $\mathcal{O}(\log(n))$ when $\alpha=2$.
  Indeed, this is a 2D version of the fact that $\sum_{i=1}^{n}\frac{1}{i^t}$ is constant if $t>1$, and $\mathcal{O}(\log n)$ if $t=1$.   
  Therefore, by~(\ref{equ:zetGtot}), the gain is constant if $\alpha>2$, and $\mathcal{O}(\log n)$ if $\alpha=2$

  We remark that~(\ref{equ:zetGtot}) also holds for the case $1<\alpha<2$.
  In this case, we have  $\zeta_\alpha\in\mathcal{O}(n^{1-\frac{\alpha}{2}})$.
  This implies that the cooperation gain is $\mathcal{O}(n^{1-\frac{\alpha}{2}})$ when $1< \alpha <2$.
  This upper bound is, however, not tight because non-cooperative broadcast can do better in grid networks when $1<\alpha<2$.
  In particular, note that the source can send the message to everyone with a single transmission of power at most $\left(\sqrt{2n}\cdot d\right)^\alpha$,
  because the maximum distance between any two nodes in the grid is $\sqrt{2n}\cdot d$. 
  When $1<\alpha<2$, this simple non-cooperative broadcast algorithm (with a single transmission) has asymptotically lower total power consumption 
  (of a factor of $\mathcal{O}(n^{1-\frac{\alpha}{2}})$) than the one described earlier
  (i.e., the one in which nodes transmit with power $d^\alpha$).
  This means that the cooperation gain in grid networks is indeed constant when $1<\alpha<2$.
  


  \begin{figure}[ht]
	\includegraphics[scale=0.5]{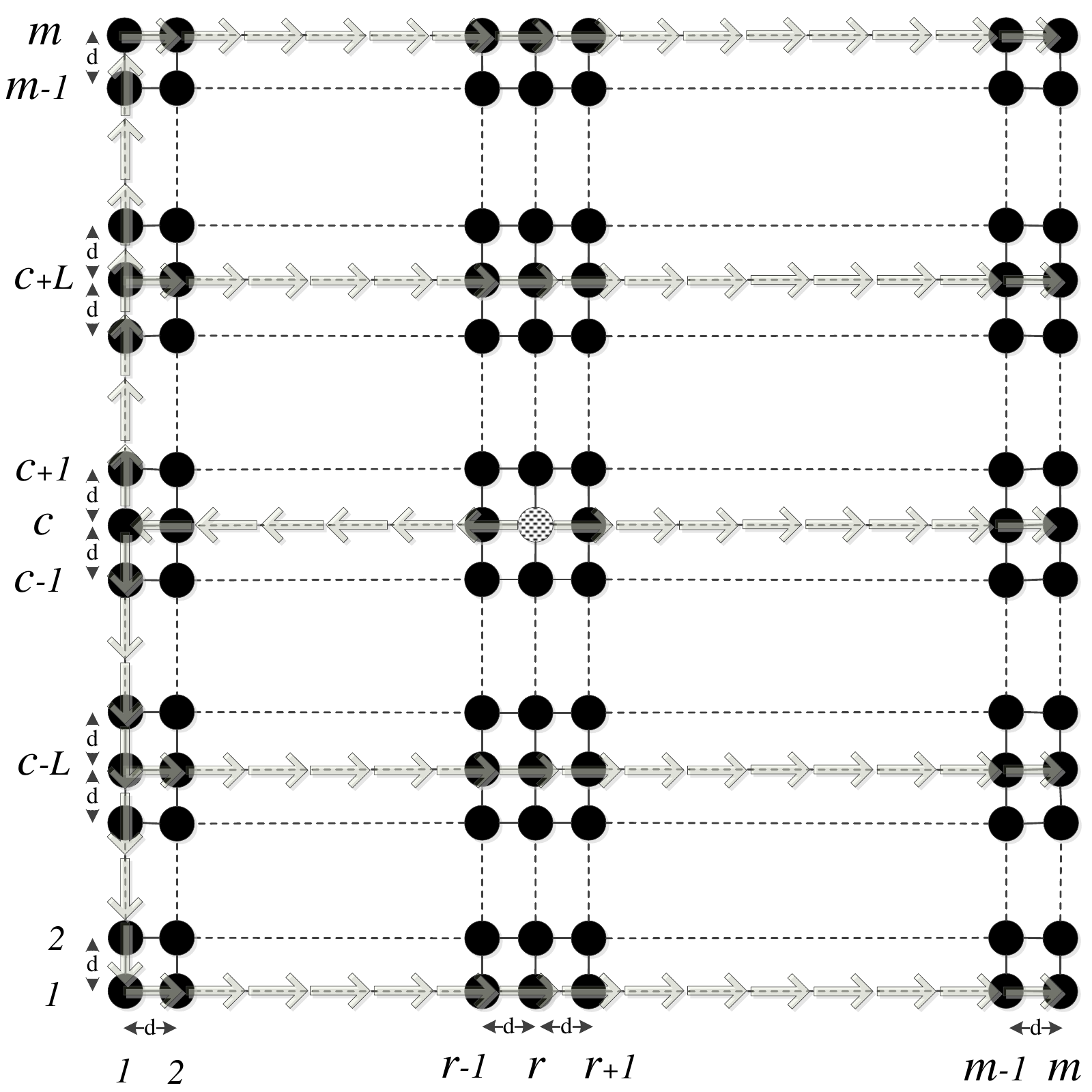}
	\centering
	\caption{An $m\times m$ grid network with  minimum node distance $d$, and $n=m^2$ nodes.}
	\label{grd:2D}
\end{figure}

  Analyzing the gain in general networks is challenging, because the analysis should work for any arbitrary arrangement of nodes in the network.
 To handle this, we first propose a \textit{novel conversion} method that converts any given cooperative broadcast algorithm into a non-cooperative algorithm.
 Then, we leverage the conversion method to establish upper bounds on the cooperation gain, $G_{tot}$, 
 achievable in 2D networks for the two cases of i) $\alpha>2$, and ii) $\alpha=2$.

%

We begin with $\alpha>2$, and show that any cooperative broadcast algorithm with power consumption $P_{tot}$ can be converted to a non-cooperative algorithm with power consumption $c \times P_{tot}$ where $c$ is a constant value, independent of the network size and topology. 
Therefore, $\Gtot = \mathcal{O}(1)$ in any 2D networks when $\alpha>2$.
After that, we consider the case $\alpha=2$,  and prove that $\Gtot = \mathcal{O}(\log n)$.
Then, we show that when $\alpha=2$, the cooperation gain indeed increases logarithmically with the number of nodes in grid networks.

\subsection{The Conversion Method}
\label{Conv-Method}
In this section, we propose a carefully crafted conversion method that converts any  
cooperative broadcast algorithm $(\mathscr{A}^{(c)},\rho^{(c)})$
into a non-cooperative broadcast algorithm $(\mathscr{A}^{(n)},\rho^{(n)})$.

Consider any 2D network with a set of $n$ nodes $\mathscr{U}$.  
Let $(\mathscr{A}^{(c)},\rho^{(c)})$ be an arbitrary cooperative broadcast algorithm with power consumption $P_{tot}^{(c)}$.
To better explain our conversion method, we will use the simple network shown in Fig.~\ref{fig:Samp1} as an example.
In Fig.~\ref{fig:Samp1}, the set of nodes transmitting in the cooperative algorithm  is $\mathscr{A}^{(c)} = \{u_1, u_3, u_5, u_6\}$.
These nodes are shown by asterisks.
The transmission order is $u_1 <^{(c)} u_3 <^{(c)} u_5 <^{(c)} u_6$ with $u_1$ being the source node.
The task of the conversion method is to 1) assign powers to the node, that is to determine the function $\rho^{(n)}$, and 2) determine the order of transmissions.  
Before explaining the power assignment and the order of transmissions, we need a few definitions.
\begin{figure}[htb!]
	\includegraphics[width = 50mm,scale = 1]{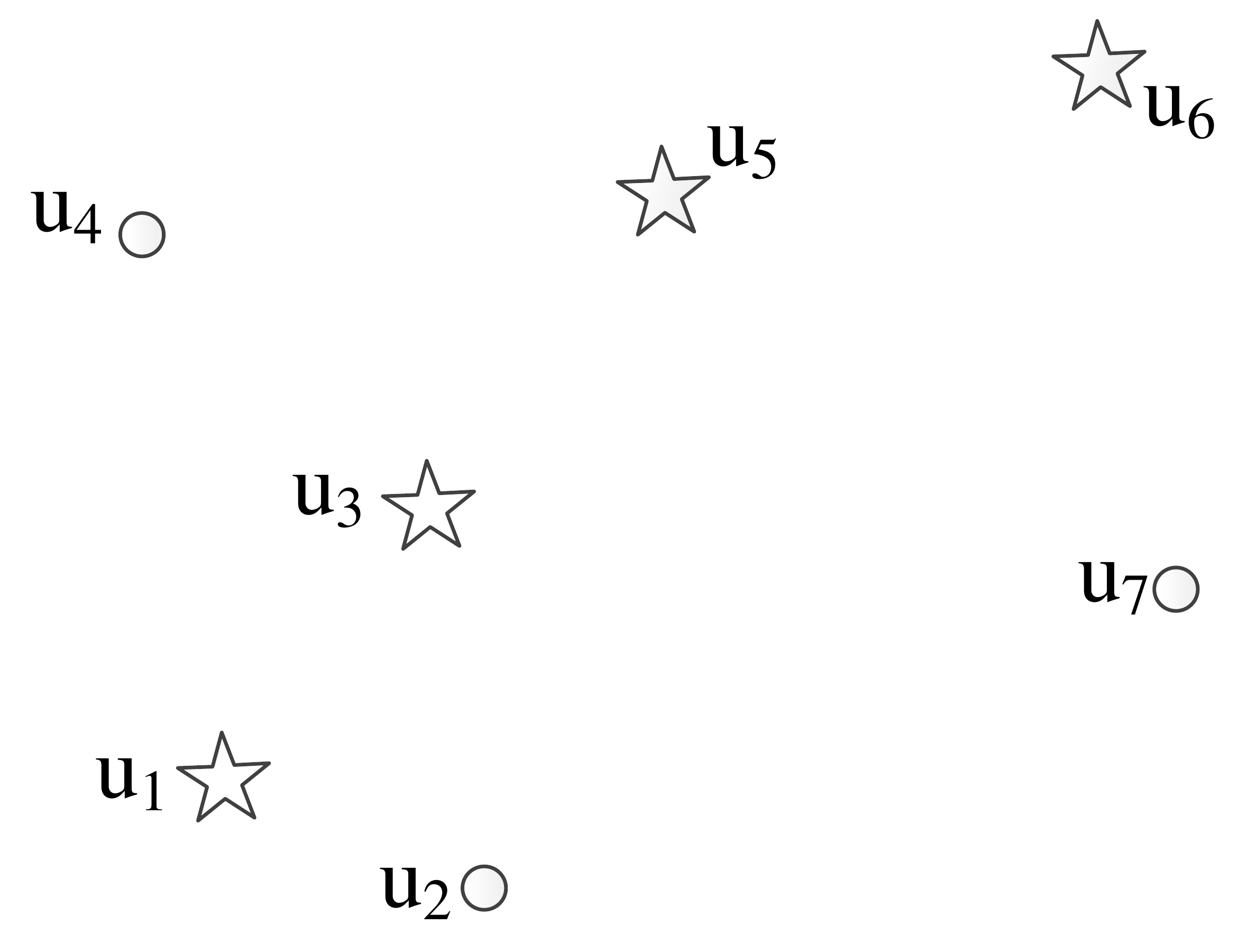}
	\centering
	\caption{The execution of $(\mathscr{A}^{(c)},\rho^{(c)})$ in a sample network. The algorithm $(\mathscr{A}^{(c)},\rho^{(c)})$ is the input to our conversion method.}
	\label{fig:Samp1}
\end{figure}

For any node $u\in\mathscr{U}\backslash\{s\}$, let $\mathcal{R}(u)$ be the closest node to $u$ that transmits before $u$ in the cooperative algorithm $(\mathscr{A}^{(c)},\rho^{(c)})$.
Formally 
\begin{equation}
	\begin{split}
		\label{equ:resp}
		\mathcal{R}:\mathscr{U}\backslash\{s\} \to \mathscr{A}^{(c)}\\
		\mathcal{R}(u) \in \argmin_{v <^{(c)} u} d_{u,v}.
	\end{split}
\end{equation}
In our sample network (Fig.~\ref{fig:Samp1}), we have
\begin{equation*}
	\begin{split}
		\mathcal{R}(u_7) = u_6, \;\; \mathcal{R}(u_6) = u_5, \;\; \mathcal{R}(u_5) = u_3,\\
		\mathcal{R}(u_4) = u_3, \;\; \mathcal{R}(u_3) = u_1,  \;\; \mathcal{R}(u_2) = u_1. 
	\end{split}
\end{equation*}
For any node $u\in\mathscr{U}\backslash\{s\}$ let
\begin{equation}
	\label{allDisks}
	D_u = \{(x,y)|(x-x_u)^2+(y-y_u)^2 \leq d_{u,\mathcal{R}(u)}^2\}
\end{equation}
be the disk with radius $r_u =d_{u,\mathcal{R}(u)}$ centred at node $u$, where $d_{x,y}$ denotes the distance between two nodes $x$ and $y$.
Let
\[
	\mathcal{D} = \{D_u| u \in  \mathscr{U}\backslash\{s\}\}.
\]
In other words, $\mathcal{D}$ is the set of every disk centred at a node $u\in\mathscr{U}\backslash\{s\}$ with the closest node that transmits before $u$ being located on the boundary of the disk. Fig.~\ref{fig:Samp2} shows the disks $D_{u} \in \mathcal{D}$ for our sample network.

\begin{figure}[htb!]
	\includegraphics[width=0.8\columnwidth]{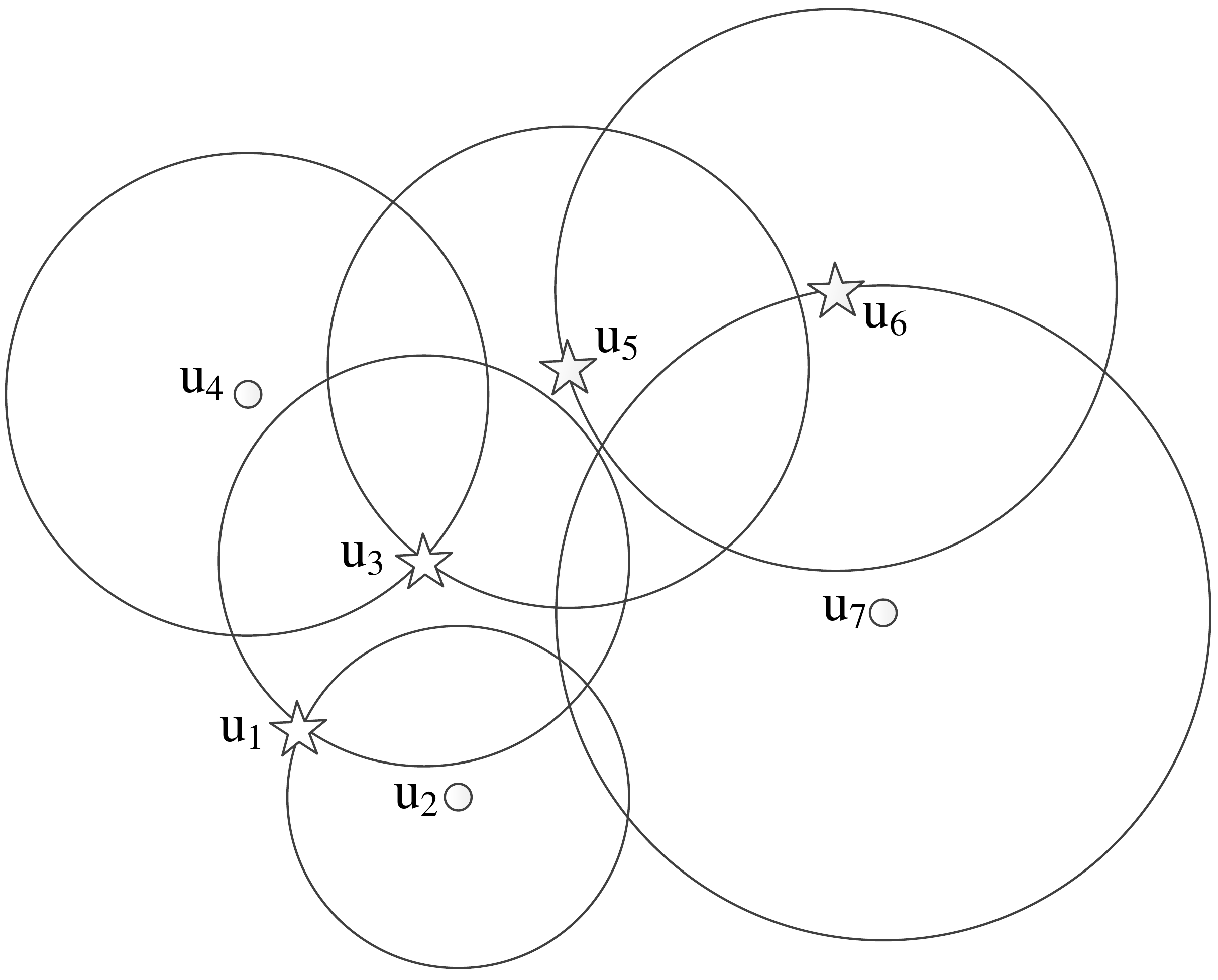}
	\centering
	\caption{Disks $D_u \in \mathcal{D}$ in the sample network.}
	\label{fig:Samp2}
\end{figure}

The power assignment is done in two steps:
\begin{enumerate}
	\item Step 1:
	\label{Step1}
	In this step, we find a subset  $\mathcal{I}$ of non-overlapping disks in $\mathcal{D}$ through an iterative process.
		Initially, we set $\mathcal{I} = \{\}$.
		In each iteration, we find the largest disk in $\mathcal{D} \backslash \mathcal{I}$ that does not overlap with any disk in $\mathcal{I}$, and add that to $\mathcal{I}$.
		We stop when every disk in $\mathcal{D} \backslash \mathcal{I}$ overlaps with at least one disk in $\mathcal{I}$.
		Disks with thick boundary in Fig.~\ref{fig:Samp3} demonstrate the set of non-overlapping disks $\mathcal{I}$ for our sample network.
	
	\item Step 2:
	\label{Step2}
	Let
		\[
		\mathcal{I} = \{D_{u_1},D_{u_2},...,D_{u_{|\mathcal{I}|}}\}, 
		\]
	be the result of the first step.
	To every disk $D_{u_i}\in \mathcal{I}$, we assign a node $w_i$ 
	\begin{equation}
		\label{equ:w_i}
		w_i = \min_{u \in \mathcal{S}_i}(u)
	\end{equation}
	where the minimization is with respect to $<^{(c)}$, and
	\begin{equation}
		\label{equ:S_i}
		\mathcal{S}_i = \{\mathcal{R}(v)|D_v \cap D_{u_i} \neq \emptyset ,r_v \leq r_{u_i} ,D_v \in \mathcal{D}\}.
	\end{equation}
	Equation~\ref{equ:w_i} simply implies that $w_i$ is the node in $\mathcal{S}_i$ that transmit before any other node in $\mathcal{S}_i$.	
	This rather strange/complex choice of node $w_i$ and the set $\mathcal{S}_i$ are better understood in the proof of Theorem~\ref{the:FD}.
	Shaded asterisks in Fig.~\ref{fig:Samp3} represent nodes $w_i$ for disks $D_{u_i} \in \mathcal{I}$ in our sample network. 
	Note that every node in $\mathcal{S}_i$ is a transmitting node because, by (\ref{equ:resp}),  $\mathcal{R}(v)\in  \mathscr{A}^{(c)}$ for every node $v\neq s$.

	In our constructed non-cooperative broadcast algorithm $(\mathscr{A}^{(n)},\rho^{(n)})$, only nodes $w_i$, $1\leq i \leq |\mathcal{I}|$, are assigned non-zero transmission power accordingly to Algorithm~\ref{alg:1}.
	
\end{enumerate}

\begin{algorithm}
	\caption{Power Assignment}
	\label{alg:1}
		\begin{algorithmic}[1]
			\State $\forall u \in \mathscr{U}$ set $\rho^{(n)}(u) \gets 0$
			\For{ $i \gets 1,~ |\mathcal{I}|$}
			\State $\rho^{(n)}(w_i) \gets \max\{\rho^{(n)}(w_i),(5r_{u_i})^\alpha\}$
			\EndFor
	\end{algorithmic}
\end{algorithm}


\begin{figure}[htb!]
	\includegraphics[width=0.8\columnwidth]{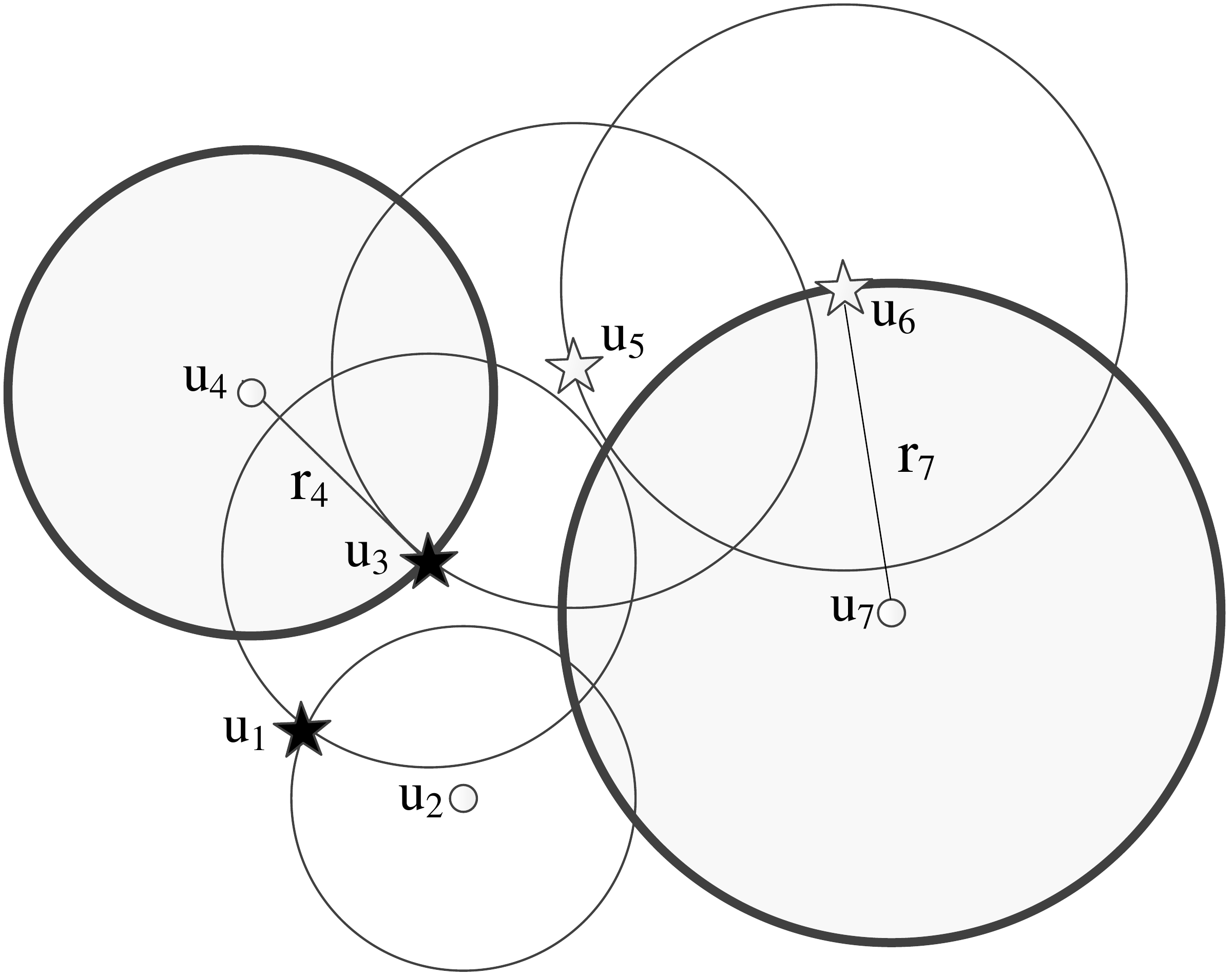}
	\centering
	\caption{Finding transmitting nodes of non-cooperative broadcast algorithm.}
	\label{fig:Samp3}
\end{figure}
\noindent
In the sample network, the transmitting set of our constructed non-cooperative algorithm is
\[
\mathscr{A}^{(n)} = \{u_1, u_3\},
\]
and the transmission powers are
\begin{equation*}
\rho^{(n)}(u_1) = (5r_{7})^\alpha, \;\;\;\;\;\rho^{(n)}(u_3) = (5r_{4})^\alpha.
\end{equation*}
Note that the constructed non-cooperative algorithm is not necessarily efficient in terms of power consumption since it is designed to alleviate the obstacles in the analysis of cooperation gain for 2D wireless networks.

\begin{remark}
\label{rem:ptotCons}
	\normalfont
	By Line 3 of Algorithm~\ref{alg:1}, we have 
	\[
		P_{tot}^{(n)} \leq \sum_{D_{u_i} \in \mathcal{I}}^{} (5r_{u_i})^\alpha,
	\]
	where $P_{tot}^{(n)}$ is the power consumption of the constructed non-cooperative algorithm.
\end{remark}

\begin{remark}
	\normalfont
	It can be inferred from Algorithm~\ref{alg:1} that $\mathscr{A}^{(n)} \subseteq \mathscr{A}^{(c)}$, since
		\[
		\forall 1 \leq i \leq |\mathcal{I}|: ~ w_i \in \mathscr{A}^{(c)}.
		\]
	i.e, every node that transmits in the constructed non-cooperative algorithm, also transmits in the given cooperative algorithm.
\end{remark}	
	With the power assignment summarized in Algorithm~\ref{alg:1}, the only remaining task to fully define our constructed non-cooperative broadcast algorithm	is to establish a transmission order among the transmitting nodes $\mathscr{A}^{(n)}$.
	To this end, in the constructed cooperative algorithm, we simply follow the same transmission order as in the non-cooperative algorithm, that is
	\begin{equation}
		\label{equ:order}
		\forall u,v \in \mathscr{A}^{(n)}:~ u <^{(c)} v \Rightarrow u <^{(n)} v,
	\end{equation}
	or equivalently
	\begin{equation}
		\label{equ:order-m}
		\forall u \in \mathscr{A}^{(n)}:~ u <^{(c)} v \Rightarrow u <^{(n)} v.
	\end{equation}

\begin{theorem}
	\label{the:FD}
	The constructed broadcast algorithm $(\mathscr{A}^{(n)},\rho^{(n)})$ is a non-cooperative broadcast algorithm.
		(i.e., the message is fully delivered to all the nodes in the network.)
\end{theorem}

\begin{proof}
  Please see Appendix~\ref{app:FD} \\
\end{proof}

So far, we have shown how to construct a non-cooperative broadcast algorithm from any cooperative broadcast algorithm.
Next, we compare the total power consumption of the cooperative algorithm with that of the constructed non-cooperative algorithm to derive upper bounds on $\Gtot$ 
for different path loss exponents~$\alpha$.

\subsection{Upper Bound Analysis for $\alpha>2$}
\label{Upper-gt-2}
In the previous section, we presented a conversion method that gets an arbitrary cooperative broadcast algorithm as input
and returns a non-cooperative broadcast algorithm as output.
In the process of this conversion, a set of non-overlapping disks 
$\mathcal{I} = \{D_{u_1},D_{u_2},...,D_{u_{|\mathcal{I}|}}\} $
was constructed.
In Remark~\ref{rem:ptotCons}, it was shown  that the total power consumption of the constructed non-cooperative algorithm, 
$P_{tot}^{(n)}$, is at most
\begin{equation}
\label{equ:ptotCons}
  P_{tot}^{(n)} \leq \sum_{D_{u_i} \in \mathcal{I}}^{} (5r_{u_i})^\alpha,
\end{equation}
where $r_{u_i}$ denotes the radius of the disk $D_{u_i}$ in the set  $\mathcal{I}$.

Our goal here is to analyze the ratio 
\[
  \frac{\sum_{D_{u_i} \in \mathcal{I}}^{} (5r_{u_i})^\alpha}{P_{tot}^{(c)}},
\]
which by (\ref{equ:ptotCons}) is an upper bound on $\frac{P_{tot}^{(n)}}{P_{tot}^{(c)}}$,
hence an upper bound on the gain of cooperative broadcast.
First, we assume that $\alpha>2$, which is the typical case in practice.
Later, in Section~\ref{Upper-eq-2}, we study the case $\alpha=2$, which corresponds to propagation in free-space.

Consider a 2D network, and let $(\mathscr{A}^{(c)},\rho^{(c)})$ be any arbitrary cooperative broadcast algorithm.
Let us call a disk in the network ``bright''\footnote{A more accurate definition of ``bright disk'' will be given later in Definition~\ref{def:BD}.}
if the sum of the received powers at any point in the disk is above the 
decoding threshold.
Let $\mathcal{I}_{\gamma}$ denote the set of all disks in $\mathcal{I}$ contracted by a factor $\gamma$, where $\gamma > 1$ is a real number. 
In other words, for every disk $D_{u_i}$ in $\mathcal{I}$, there is a disk $D'_{u_i}$ in $\mathcal{I}_{\gamma}$
that has the same center as $D_{u_i}$ but $1/\gamma$ fraction of its radius.

Before delving into the complicated details of analyzing the upper bound, 
let us provide a roadmap to our analysis.

\noindent
\textbf{Roadmap:}
\begin{enumerate}
  \item We multiply all transmission powers in the cooperative broadcast algorithm $(\mathscr{A}^{(c)},\rho^{(c)})$  by a constant number (to be set later).
  \item We show that, now,  every disk in $\mathcal{I}_{\gamma}$ is bright. 
  \item In a carefully crafted iterative process, we transfer transmission powers from nodes in $\mathscr{A}^{(c)}$ 
    to the centers of the disks in $\mathcal{I}_{\gamma}$.
  \item We prove that after all these power transfers, every disk in $\mathcal{I}_{\gamma}$ will remain bright.
  \item 
    We show that if the transmission sources are placed at the center of disks, then to make all the disks in $\mathcal{I}_{\gamma}$ bright, we need 
    $\theta(\sum_{i=1}^{|\mathcal{I}_{\gamma}|} {r'}_{u_i}^\alpha) $ total power, where ${r'}_{u_i}$, $1\leq i \leq |\mathcal{I}_{\gamma}|$ denote the radii  of disks in $\mathcal{I}_{\gamma}$.
    
\end{enumerate}

In the above process, the total transmission power of $(\mathscr{A}^{(c)},\rho^{(c)})$  changes only by a constant in the first step.
In the remaining steps, the total transmission power does not change; In Step 3, we only transfer some powers from one point to another point.
By the last item of the roadmap, the total transmission power is $\theta(\sum_{i=1}^{|\mathcal{I}_{\gamma}|} {r'}_{u_i}^\alpha)$.
Therefore, by~(\ref{equ:ptotCons}) and the fact that ${r'}_{u_i}=\gamma r_{u_i}$ we get $\frac{P_{tot}^{(n)}}{P_{tot}^{(c)}}=\mathcal{O}(1)$.



Now, we explain each step of the roadmap in details.
In the first step of the roadmap, we multiplying the transmission powers of $(\mathscr{A}^{(c)},\rho^{(c)})$  by $(1+1/\gamma)^{\alpha}$.
In other words, we set 
\[
	\forall u \in \mathscr{A}^{(c)}:~~~~ \rho^{(c)}(u)  \longleftarrow (1+1/\gamma)^{\alpha} \times \rho^{(c)}(u).	
\] 
Note that $(1+1/\gamma)^{\alpha}$ is constant (irrespective of the value of~$\gamma$) because $1+1/\gamma < 2$ as $\gamma >1$.

The next step, as stated in the roadmap, is to prove that for any $\gamma>1$, every disk in $\mathcal{I}_{\gamma}$ is ``bright''.
Before doing so, let us precisely define the term ``bright disk''.
\begin{definition}[Bright Disk]
\label{def:BD}
  Let $D'_{u_i}$ be a disk  in $\mathcal{I}_{\gamma}$, where $\gamma>1$.
  We say that the disk $D'_{u_i}$ is \emph{bright} with respect to a set of transmitters $\mathscr{T}$ 
  if at any point in the disk $D'_{u_i}$, the sum of received powers from transmitters in $\mathscr{T}$  
  is above the decoding threshold.
  The term ``bright'' comes from viewing the transmitters as sources of light (e.g., light bulbs).
\end{definition}

Recall that $\mathscr{A}^{(c)}$ denotes the set of transmitters.
For any disk $D_{u_i} \in \mathcal{I}$, 
let $\mathscr{T}_{u_i}$ denote the set of nodes in $\mathscr{A}^{(c)}$ that are not inside the disk $D_{u_i}$.
By the definition of $D_{u_i}$, every node that transmits before $u_i$ is not inside the disk $D_{u_i}$, thus it will be in the set $\mathscr{T}_{u_i}$.


\begin{prop}
  \label{prp:bri}
  Every disk $D'_{u_i}\in \mathcal{I}_{\gamma}$ is bright with respect to $\mathscr{T}_{u_i}$.
\end{prop}
\begin{proof}
  Please see Appendix~\ref{app:bri}.
\end{proof}

Next step, as mentioned in the roadmap (item 3), is to transfer powers from transmitting nodes in $\mathscr{A}^{(c)}$ 
to nodes at the center of the disks in $\mathcal{I}$.
We perform these power transfers iteratively as follows.


\vspace{0.2in}
\noindent
\textbf{The Iterative Power Transfer Procedure:}
  \begin{enumerate}
    \item \textbf{Offering powers:}\\
      For every $v \in \mathscr{A}^{(c)}$, let 
      \[
    		\mathcal{C}_v = \{u_i | D_{u_i} \in \mathcal{I},~~  v \in \mathscr{T}_{u_i}\}.
    	  \]  

      In the $j$th iteration, $j\geq 1$, each node $v \in \mathscr{A}^{(c)}$ offers all its transmission power remaining from the previous iterations to
      its $j$th closest node in $\mathcal{C}_v$.
      A node in $\mathcal{C}_v$ may be offered powers from multiple nodes in each iteration, and 
      may take any portion of the powers offered to it as will be explained in Step 2 (taking powers).
    \item \textbf{Taking powers:}\\
      Imagine that each node $u_i$ at the center of disk $D_{u_i}\in\mathcal{I}$ has a bucket to be filled by the offered powers.
      Initially all buckets are empty, and the size of the bucket assigned to node $u_i$ is set to $(\frac{r_{u_i}}{2})^\alpha$.
      In each iteration, $u_i$  is offered powers from a subset of nodes in $\mathscr{A}^{(c)}$,
      according to Step~1.
      In response, $u_i$ takes as much as those offered powers (from arbitrary offerers, and of arbitrary portions) to make its bucket full.
      If $u_i$'s bucket becomes full, $u_i$ will return any extra power offered to its offerer,
      and will decline any offer in the next iterations.
  \end{enumerate}
  
For every $v \in \mathscr{A}^{(c)}$, the size $|\mathcal{C}_v| \leq |\mathcal{I}|$. 
Therefore, a node $v \in \mathscr{A}^{(c)}$ can offer its power in at most $|\mathcal{I}|$ iterations, 
hence the above power transfer procedure will eventually terminate.
Clearly, when the power transfer procedure terminates, we will have one of the following two scenarios 
1) all the buckets are full; 2) there is at leas one bucket that is not full.
We first cover the first scenario, which is the simpler scenario in our analysis:

\bigskip
\noindent
\textbf{Scenario 1:}\\
  Note that in the power transfer procedure, powers from nodes in $\mathscr{A}^{(c)}$ are transferred into buckets.
  If all buckets are full, the total transmission powers at nodes in  $\mathscr{A}^{(c)}$ (i.e., $3^{\alpha} P_{tot}^{c}$) must be 
  more that the sum of the sizes of all buckets, that is
  \[
    3^{\alpha} \cdot P_{tot}^{c} \geq \sum_{D_{u_i}\in \mathcal{I}} 
    \underbrace{\left(\frac{r_{u_i}}{2}\right)^{\alpha}}_{\text{size of } u_i\text{'s bucket}}.
  \]
  Therefore,
  \[
    \frac{ \sum_{D_{u_i}\in \mathcal{I}} r_{u_i}^\alpha}{P_{tot}^{c}} \leq 6^\alpha,
  \]
  Thus, by (\ref{equ:ptotCons}), the cooperative broadcast gain is at most $30^\alpha$, which is constant, 
  as $\alpha$ is a constant.

\bigskip
\noindent  
\textbf{Scenario 2:}\\
  In this scenario, all the transmission powers must have been transferred to the center of disks in $\mathcal{I}_{\gamma}$, 
  and each center $u_i$ has at most $\left( \frac{{r'}_{u_i}}{2}\right) ^\alpha$ power, which is the size of its bucket.
  To study this scenario, we analyze the following related problem:

\begin{quote}
  \textbf{Brightening Non-overlapping Disks:}\\
  One input/parameter of this problem is a set of $n$ no-overlapping disks $D_1, D_2, \ldots, D_n$, with radii $r_1, r_2, \ldots, r_n$, and centers $c_1, c_2, \ldots, c_n$, respectively. 
  The second parameter of the problem is a real number  $\gamma >1$.
  
  Let $\dist(x,y)$ denote the distance between two points $x$ and $y$.
  Suppose that $\dist(c_i, c_j)\geq \gamma (r_i+r_j)$, for every $1\leq i,j\leq n$, $i\neq j$.
  Therefore, every two different disks $D_i$ and $D_j$ are not only non-overlapping, but also separated by a gap set by the parameter $\gamma$. 
   
  Each center $c_i$, $1\leq i \leq n$, is assigned a transmission power $p_i$, where $0\leq p_i\leq r_i^\alpha$.
  Such a power assignment is called ``valid'' if it assures every disk is bright with respected to the set of all transmitters $c_j$, $1\leq j \leq n$.
  The brightening non-overlapping disk problem is to find a valid power assignment with minimum total power, 
  that is to find a valid power assignment with minimum $\sum_{i=1}^{n}p_i$.
\end{quote}

A simple power assignment valid for any parameter in the  brightening non-overlapping disks problem is $p_i=r_i^\alpha$.
The next theorem proves that this simple power assignment is asymptotically optimal when $\alpha>2$, and $\gamma$ is large enough.

\begin{theorem}
\label{thm:valid}
    In the brightening non-overlapping problem with parameter $\gamma \geq g(\alpha)$, for any valid power assignment $p_i$, $1\leq i \leq n$, we have
  \[
   \sum_{i=1}^{n} p_i \geq \beta \cdot \sum_{i=1}^{n} r_i^\alpha,
  \]
  where
  \[
    \beta = 1- \left( \left( \frac{\gamma}{\gamma-1} \right)^\alpha +1 \right) \left( \frac{1}{(\alpha-2) \cdot \gamma^{\alpha-2}}\cdot\left(\frac{2\gamma+1}{2\gamma-1}\right)^\alpha \right)
  \]
\end{theorem}
\begin{proof}
  Please see Appendix~\ref{app:valid}.
\end{proof}

The above theorem shows that we need at least the total power of 
$\beta\cdot  \sum_{D_{u_i}\in \mathcal{I}} \left(\frac{r_{u_i}}{\gamma}\right)^\alpha$ to brighten all disks in $\mathcal{I}_{\gamma}$.
Therefore, following the same argument given in Scenario 1, we get that the cooperation gain is $\mathcal{O}(1)$.

\begin{example}
Suppose the path loss exponent $\alpha=3$.
Let us set $\gamma=20$.
Then, using the definition of $\beta$ in the statement of Theorem~\ref{thm:valid}, we get $\beta\geq \frac{1}{3}$.
This implies that, we need at least $\frac{1}{3}\cdot  \sum_{D_{u_i}\in \mathcal{I}} \left(\frac{r_{u_i}}{20}\right)^3$ total power to brighten all the disks in $\mathcal{I}_{20}$.
On the other hand, by Lemma~\ref{prp:bri}, we know that all the disks in $\mathcal{I}_{20}$ can be brightened 
using the total power of $(1+1/20)^3 \cdot P_{tot}^c$.
Thus, we must have
  \[
    (1+1/20)^3 \cdot P_{tot}^c \geq \frac{1}{3}\cdot  \sum_{D_{u_i}\in \mathcal{I}} \left(\frac{r_{u_i}}{20}\right)^3
  \]
Therefore, irrespective of value of $n$
  \[
    \frac{ \sum_{D_{u_i}\in \mathcal{I}} r_{u_i}^\alpha}{P_{tot}^{c}} \leq 3.48\cdot20^3,
  \]
 Thus, by (\ref{equ:ptotCons}), the cooperation gain is constant.
\end{example}

\subsection{Upper Bound Analysis for $\alpha=2$}
\label{Upper-eq-2}

Next theorem (Theorem~\ref{the:consRatio}) proves that the total power consumption of the constructed non-cooperative algorithm is 
at most $\mathcal{O}(\log n)$ times that of the cooperative algorithm used for the conversion.
This implies that the cooperation gain, $G_{tot}$, is in $\mathcal{O}(\log n)$ 
because any given cooperative algorithm can be converted to a non-cooperative algorithm with at most a factor of $\mathcal{O}(\log n)$ increase in the total transmission power.

\begin{theorem}
	\label{the:consRatio}
	Let $(\mathscr{A}^{(c)}, \rho^{(c)})$ be a cooperative broadcast algorithm, and $(\mathscr{A}^{(n)}, \rho^{(n)})$ be the non-cooperative broadcast algorithm constructed from it. 
	We have
	\[
	G_{tot} = \frac{\sum_{u_i \in \mathscr{U}}^{}\rho^{(n)}(u_i)}{\sum_{u_i \in \mathscr{U}}^{}\rho^{(c)}(u_i)} \in \mathcal{O}(\log n).
	\]\end{theorem}
\begin{proof}
  Please see Appendix~\ref{app:consRatio}.
\end{proof}

Theorem~\ref{the:consRatio} shows the the cooperation gain grows at most logarithmically with the number of nodes $n$.
An interesting question is whether a logarithmic growth can be achieved in any 2D networks.
Following, we show that in grid networks the growth of cooperation gain is indeed logarithmic in $n$.

\bigskip
\subsubsection{Cooperation Gain in Grid Networks}
\label{Lower-B-A}
\hfil \break
Fig. \ref{grd:2D} shows the topology of a grid network with minimum node distance $d$, and $n=m^2$ nodes, where $m$ is a positive integer.

The algorithm $(\mathscr{U},\rho^{(n)}(u)=d^2)$ is a simple non-cooperative broadcast algorithm in which all nodes in $\mathscr{U}$ transmit with power $d^2$.
The total power consumption of this algorithm is clearly $P^{(n)}_{tot}=nd^2$.
There are non-cooperative broadcast algorithms with total power consumption less than $nd^2$. 
For example, if only nodes in every third row transmit,\footnote{To guarantee that every node receives the message, nodes in the top and bottom rows may need to transmit too.} 
every node will receive the message and we get $P^{(n)}_{tot} \simeq \frac{m^2d^2}{3} = \frac{nd^2}{3}$.
The next proposition, however, shows that the total power consumption of any non-cooperative algorithm for the given grid network is $\Omega(nd^2)$. This implies that $(\mathscr{U},\rho^{(n)}(u)=d^2)$ is asymptotically optimum (i.e., it can be improved by at most a constant factor).

\begin{prop}
\label{prp:grd1}
	For any non-cooperative broadcast algorithm over the 2D grid network (Fig. \ref{grd:2D}) with power consumption $P^{(n)}_{tot}$ we have 
	\[
	P^{(n)}_{tot} \in \Omega(nd^2).
	\]
\end{prop}
\begin{proof}
  Please see Appendix~\ref{app:grd1}.
\end{proof}

Next, we propose a simple cooperative broadcast algorithm \mbox{$(\mathscr{A}^{(c)},\rho^{(c)})$} with total power consumption $P_{tot}^{(c)}\in\mathcal{O}(\frac{nd^2}{\log n})$.  
Fig. \ref{grd:2D} illustrates this cooperative algorithm in which the transmission proceeds horizontally, from the source, in both directions. 
Reaching the left-most column, the broadcast continues vertically, in both directions. 
As depicted in Fig. \ref{grd:2D}, nodes in every $L$th row transmit the message, where $L$ is an adjustable parameter of the algorithm. 
Every such row is  called a \emph{transmitting row}.
The transmission power of each node $u \in \mathscr{A}^{(n)}$ is set to $d^2$.
The following proposition shows that the algorithm illustrated in Fig. \ref{grd:2D} is a cooperative broadcast algorithm (i.e. every node will successfully decode the message) even when $L$ is as large as $0.15\ln (n)$. 
This implies that total power consumption $P_{tot}^{(c)}$ can be as low as $\mathcal{O}(\frac{nd^2}{\log n})$.

\begin{prop}
  \label{prp:grd2}
	The broadcast algorithm $(\mathscr{A}^{(c)},\rho^{(c)})$ illustrated in Fig.~\ref{grd:2D} is a cooperative broadcast algorithm (i.e the message is delivered to all nodes) for any $L \leq 0.15\ln (n)$.
\end{prop}
\begin{proof}
  Please see Appendix~\ref{app:grd2}.
\end{proof}

Considering the grid network (Fig. \ref{grd:2D}), there are $\frac{m}{L}$ transmitting rows within each $m$ nodes broadcast with power $d^2$. Therefore, we have \mbox{$P_{tot}^{(c)}\in\mathcal{O}(\frac{nd^2}{\log n}$)},
hence, $G_{tot}\in\Omega(\log n)$.

\section{NUMERICAL RESULTS}\label{third:simulation}

To calculate the cooperation gain, we used the best existing broadcast algorithms in the literature.
For non-cooperative broadcast, we implemented the algorithm proposed by Caragiannis et al. in \cite{caragiannis2013exponential}.
Among the existing non-cooperative broadcast algorithms, 
this algorithm has the best approximation factor to the broadcast power consumption minimization problem.
As for the cooperative algorithm, we implemented the greedy filling algorithm proposed in \cite{maric2004cooperative}.

In the simulation, a set of $n$ nodes are placed uniformly at random in a disk.
For each given number of nodes $n$, we execute the above two algorithms over 50 different node placements.
Figures~\ref{fg:linear_noFading} shows the numerical results.

\begin{figure}[htb!]
	\includegraphics[width=0.9\columnwidth]{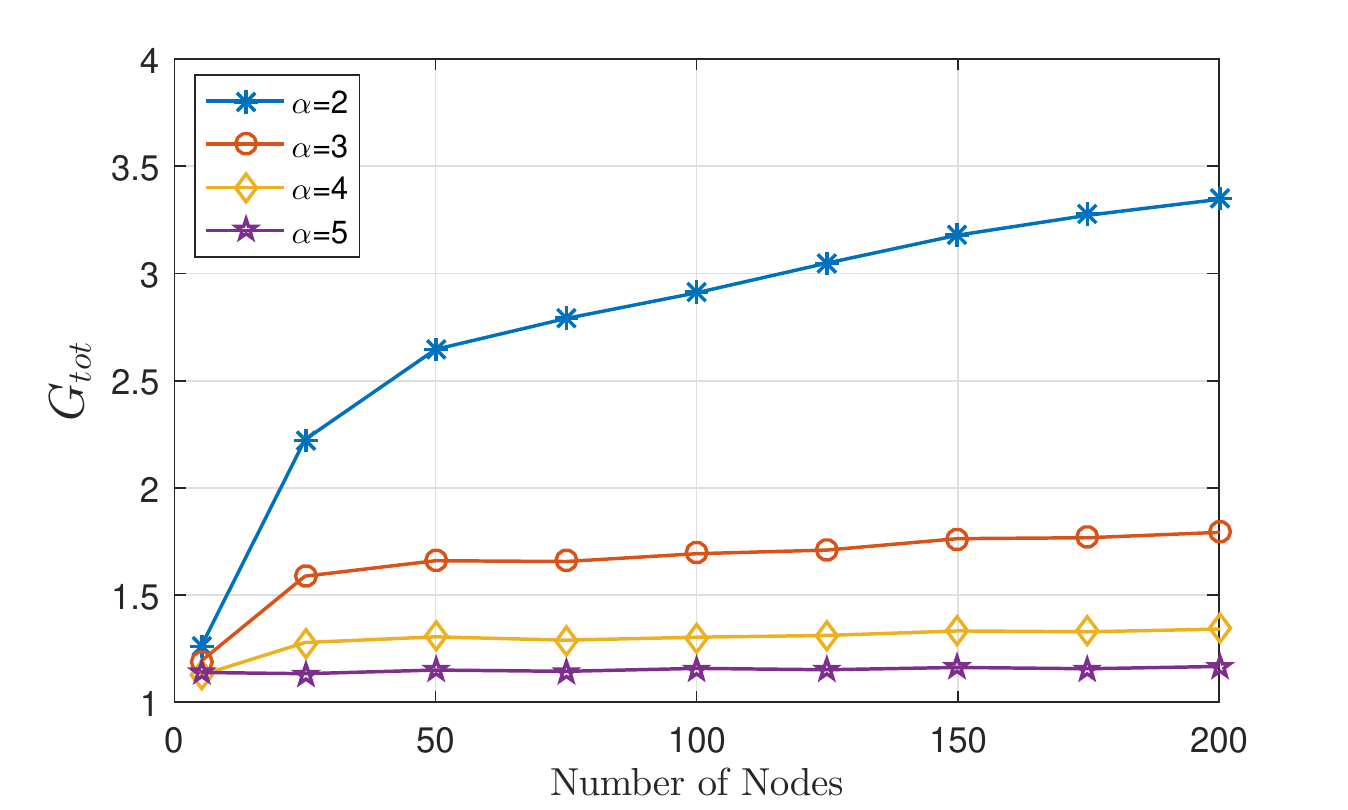}
	\centering
	\caption{$\Gtot$ versus $n$. Nodes are uniformly distributed in the network.}
	\label{fg:linear_noFading}
\end{figure}


We compared the algorithm by Caragiannis et al.~\cite{caragiannis2013exponential} with the \textit{greedy filling} algorithm under two other settings each with different node distributions.
In one setting, nodes positions follow a Gaussian distribution with standard deviation of 0.5 around the center the disk.
The result of this simulation is shown in Figures~\ref{fg:AC_linear_noFading}.

\begin{figure}[htb!]
	\includegraphics[width=0.9\columnwidth]{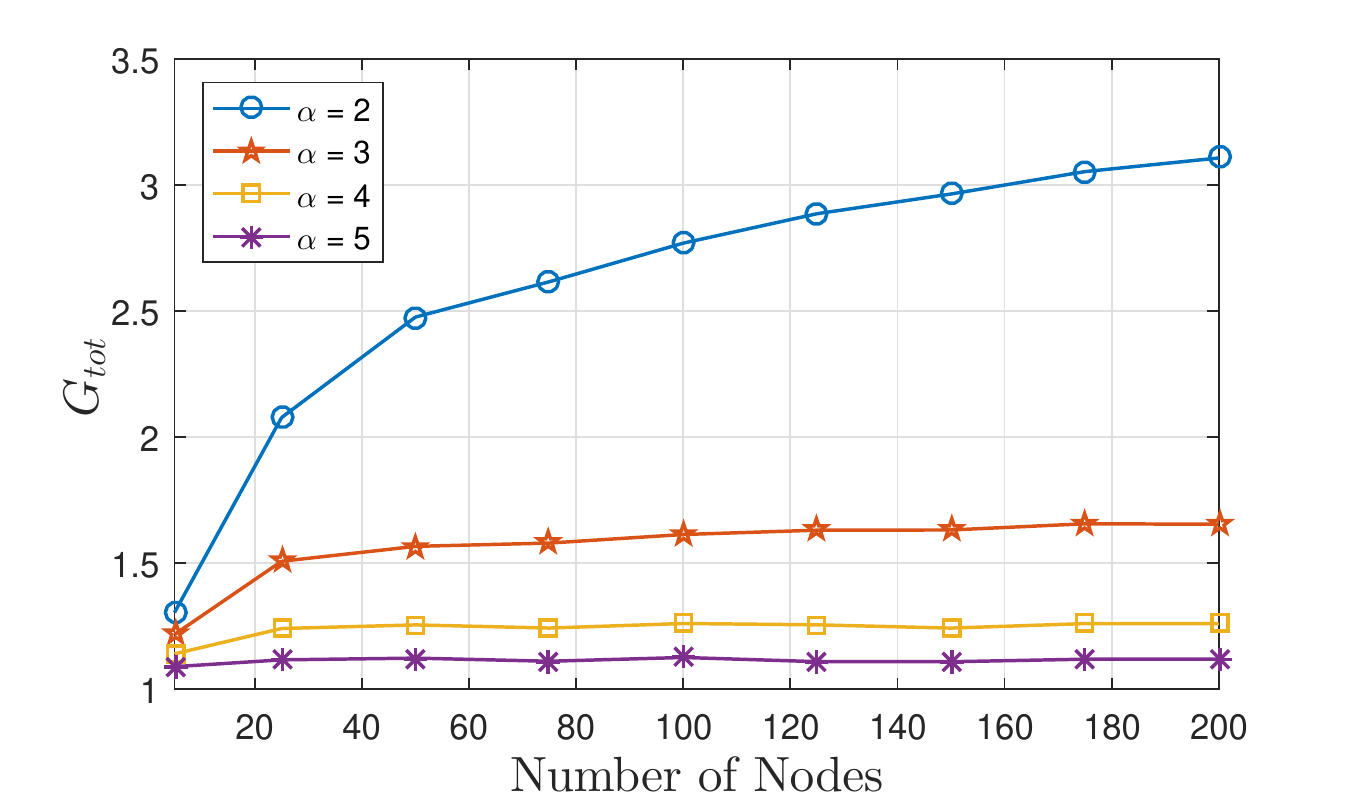}
	\centering
	\caption{$\Gtot$ versus $n$. Nodes have Gaussian distribution around center.}
	\label{fg:AC_linear_noFading}
\end{figure}

In the second setting,  for nodes distribution, we considered a clustered structure consisting of 5 cluster centers. 
Nodes locations around these centers follow a Gaussian distribution ($\sigma = 0.5$).
The numerical results are shown in Figures~\ref{fg:CL_linear_noFading}.
\begin{figure}[htb!]
	\includegraphics[width=0.9\columnwidth]{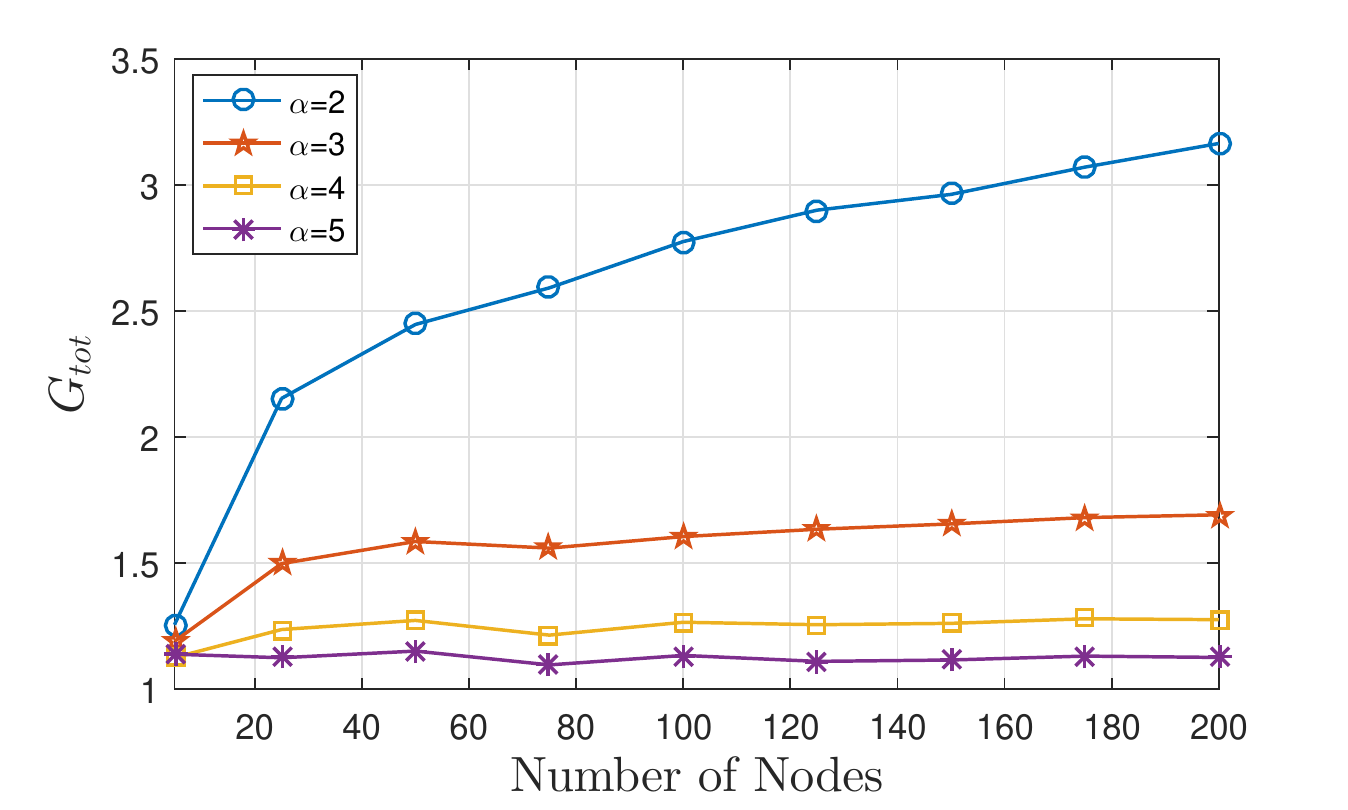}
	\centering
	\caption{$\Gtot$ versus $n$. Nodes are placed in a clustered structure.}
	\label{fg:CL_linear_noFading}
\end{figure}

Next, we evaluated our conversion method explained in Section \ref{Upper-eq-2}.
When $\alpha=2$, this method converts any given cooperative broadcast algorithm with total power consumption of $P_{tot}^{(c)}$ to 
a non-cooperative broadcast algorithm with total power consumption $P_{tot}^{(n)}\leq 127 \ln(n)$. 
To verify this result, we considered networks with size up to  $n=400 \times 400$.
We placed  nodes  uniformly at random in the network, 
and used the greedy filling algorithm proposed in \cite{maric2004cooperative} as the input to our conversion method.
For each value of $n$, we performed 1000 runs of simulations.
In every run, we verified that the constructed algorithm achieves full delivery. 
Fig.~\ref{fg:c-method-log} shows the conversion radio versus  $\ln(n)$, where the conversion  ratio is defined as 
the ratio of the total power consumption of the constructed algorithm to that of the conversion method input (in this case, the greedy filling algorithm).
The maximum slope of the curve in Fig.~\ref{fg:c-method-log} is about five, 
much lower than the slope of 127 in the proven bound of  $127 \ln(n)$ on the conversion ratio (See proof of Theorem~\ref{the:consRatio}).

\begin{figure}[htb!]
	\includegraphics[width=0.9\columnwidth]{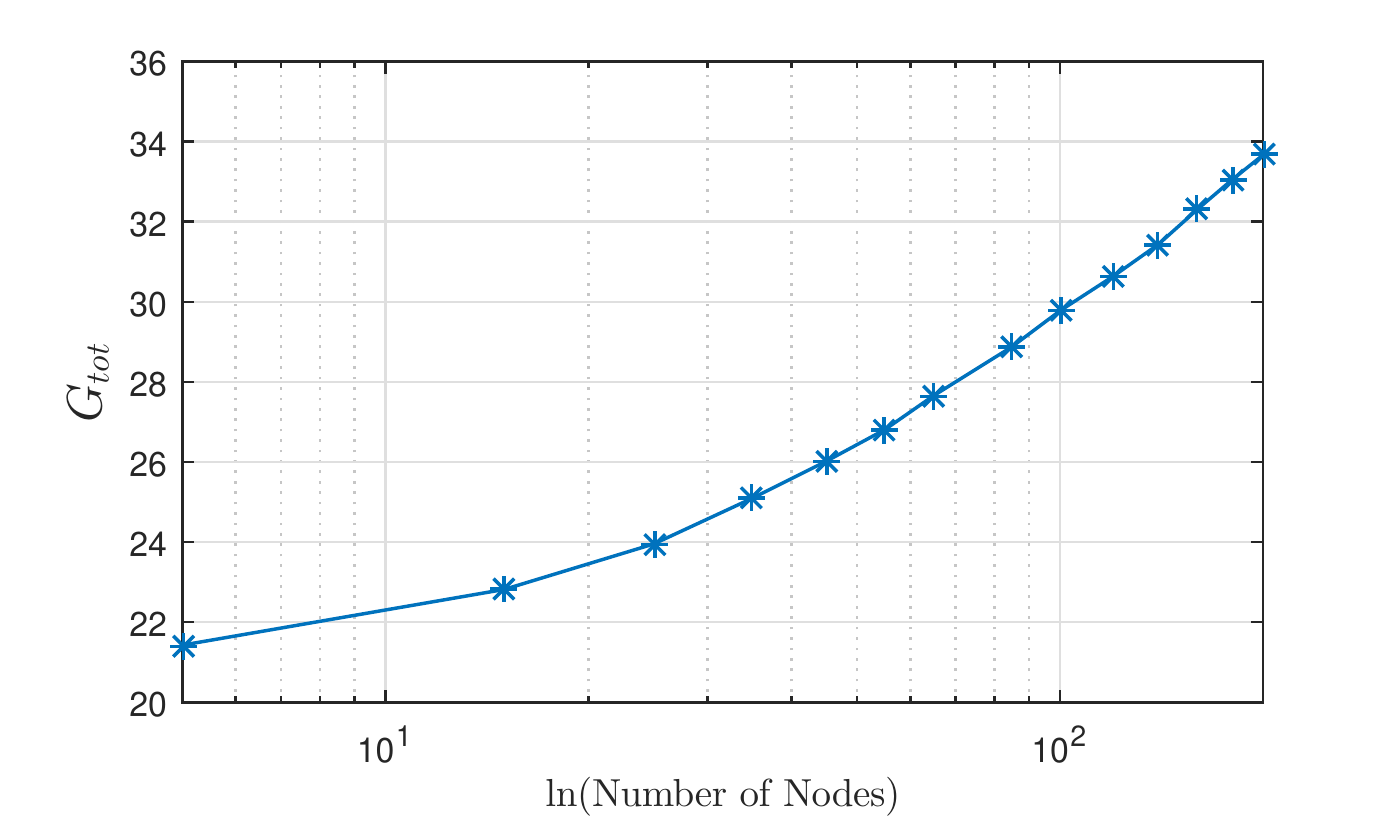}
	\centering
	\caption{The conversion ratio when the greedy algorithm is used as the input of the conversion method.}
	\label{fg:c-method-log}
\end{figure}


For grid networks, we simulated the cooperative broadcast algorithm discussed in Section~\ref{Lower-B-A}.
In the simulated algorithm, in addition to nodes in every $L$th row, the nodes in the top and bottom rows transmit.
For a given size of the grid network and a given position of the source node, we first search for the largest value of $L$ that guarantees full delivery.
Then, we use the maximum value of $L$ in the cooperative algorithm, and calculate the total power consumption.
As for the non-cooperative algorithm, we use the simple algorithm in which all the nodes transmit with power $d^2$,
where $d$ denotes the minimum distance between nodes in the grid.
This simple non-cooperative algorithm was proven to be asymptotically optimal.

In the simulation, we go up to a grid size of $50\times50$, with $n=50\times50$ nodes.
For a given grid size, we run the simulation 100 times.
In each run, the source is selected uniformly at random from the nodes in the gird.
Fig. \ref{fg:grid-ln} shows the ratio of the total power consumption of the non-cooperative algorithm over that of the cooperative algorithm.
To show the logarithmic growth of the cooperation gain, in Fig. \ref{fg:grid-ln}, we plotted the gain versus $\ln n$.

\begin{figure}[htb!]
	\includegraphics[width=0.9\columnwidth]{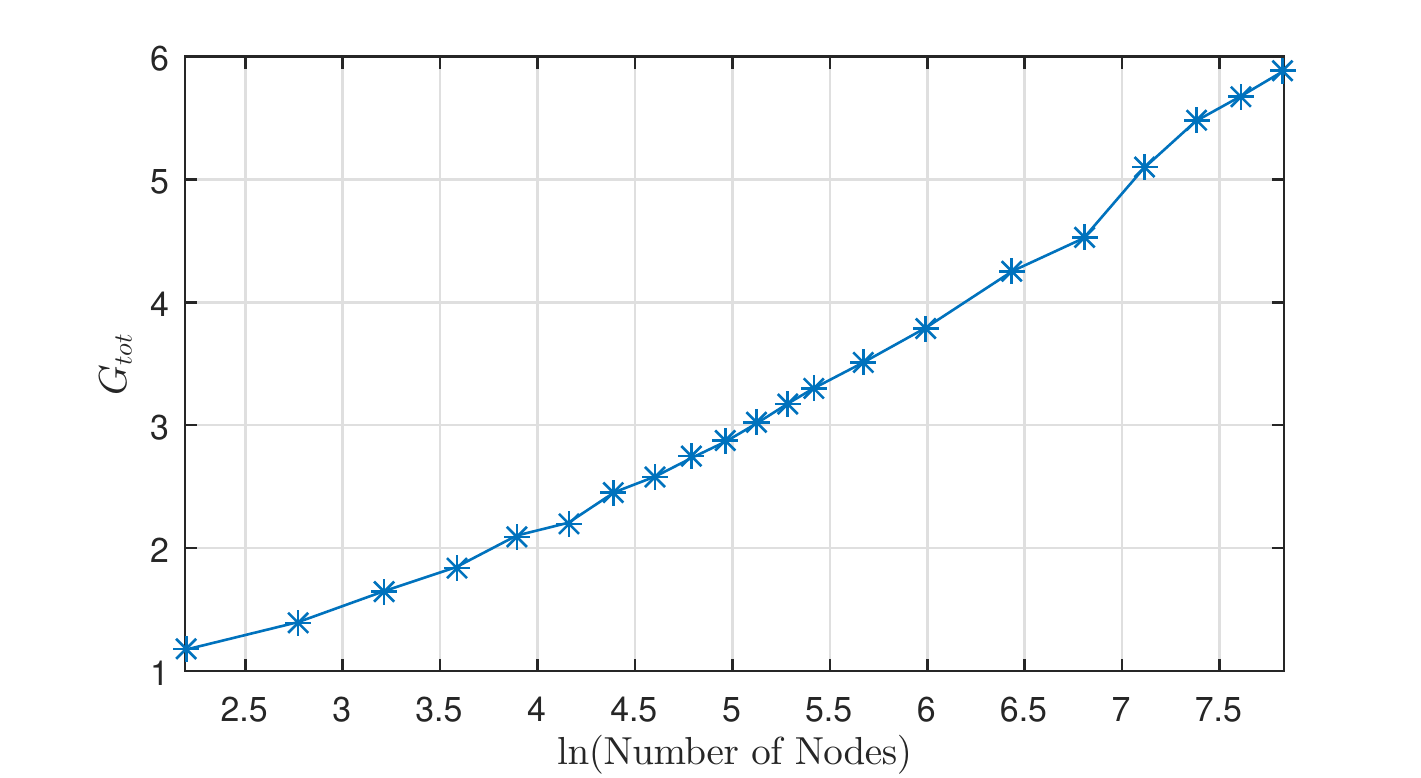}
	\centering
	\caption{$\Gtot$ versus $\ln(n)$ for 2D grid networks.}
	\label{fg:grid-ln}
\end{figure}


\section{Conclusion and Future Work}
\label{conclu}
The cooperation gain is a measure of how much energy can be saved when energy accumulation is used in wireless broadcast.
In this work, we analyzed the cooperation gain in 2D networks.
We proved that when $\alpha>2$, the cooperation gain is constant irrespective of the network size and topology.
When $\alpha=2$, we proved that the cooperation gain grows at best logarithmically with the number of nodes.
Then, we showed that in grid networks the cooperation gain indeed increases logarithmically with the number of nodes when $\alpha=2$.

This work can be extended in multiple ways. 
One is to consider random channels. 
To this end, we can use, for example, the Rayleigh fading model instead of the path loss model.
Also, in this work, we showed that the cooperation gain is constant in grid networks when $1<\alpha<2$. 
As a future work, this result can be extended to general 2D networks.
Finally, we analyzed the cooperation gain without placing any constraint on latency.
It would be interesting to know how much gain we can get from energy accumulation when we need to meet a desired latency constraint.

\begin{appendices}

\section{Proof of Theorem~\ref{the:FD}}
\label{app:FD}
	Towards showing a contradiction, we assume that there is at least one node that does not receive the message. 
	Among nodes that have not received the message, let $f$ be the smallest node with respect to the ordering $<^{(n)}$ (i.e. every node $u<^{(n)}f$ has successfully received the message).
	In the~following, we will prove that there is a node $v$, $v<f$, which has received the message and transmitted with power $\rho^{(n)}(v) \geq d^2_{v,f}$.
	This implies that node $f$ must have received the message from node $v$, a contradiction.
	We consider two cases.
	In the first case, we assume $D_f \in \mathcal{I}$; in the second case we assume $D_f \notin \mathcal{I}$.
	\begin{enumerate}
		\item Case 1: $D_f \in \mathcal{I}$:\\
			Since  $D_f \in \mathcal{I}$, we assume that $f = u_i$, for some $i \in \{1,2...,|\mathcal{I}|\}$.
			Fig. \ref{fig:Df-remain} shows disk $D_{u_i} \in \mathcal{I}$ (with thick boundary) as well as all disks in $\mathcal{D}$ that are not bigger than $D_{u_i}$ and intersect with $D_{u_i}$.
			For every disk $D_u$ shown in Fig.~\ref{fig:Df-remain}, node $\mathcal{R}(u)$ is represented by an asterisk.
			Therefore, the asterisks shown in Fig.~\ref{fig:Df-remain} are the set $\mathcal{S}_i$ defined in (\ref{equ:S_i}).
			Note that, $w_i$, defined in (\ref{equ:w_i}), is one of those asterisks. 
			We have the following for $w_i$:
			\begin{enumerate} [label=(\roman*)]
				\item $\rho^{(n)}(w_i) \geq (5r_{u_i})^\alpha$, where $r_{u_i}$ is the radius of the disk $D_{u_i}=D_f$: \\
				It is because, in the $i$th iteration of the power assignment algorithm (Algorithm~\ref{alg:1}), 
				the power of node $w_i$ (i.e. $\rho^{(n)}(w_i)$) is set to a number at least equal to $(5r_{u_i})^\alpha$. 
				Also if $w_i=w_j$ for some integer $j>i$, $\rho^{(n)}(w_i)$ will not be reduced in iteration $j$.
				\item $d_{w_i,u_i} \leq 3r_{u_i}$:\\
				Let $P$ be a point on the circumference of any disk shown in Fig.~\ref{fig:Df-remain}. 
				The distance between $P$ and $f=u_i$ is at most $3r_{u_i}$ because the radius of every disk is at most equal to $r_{u_i}$.
				\item $w_i <^{(n)} u_i$:\\
				From (\ref{equ:resp}), we get $\mathcal{R}(u_i) <^{(c)} u_i$, and by (\ref{equ:w_i}) we get $w_i <^{(c)} \mathcal{R}(u_i)$. 
				Therefore, from (\ref{equ:order-m}), we have $w_i <^{(n)} u_i$, since $w_i \in \mathscr{A}^{(n)}$.
				\item $w_i$ has received the message:\\ 
				This is simply by the assumption that every node $u<^{(n)} f$ has received the message and because $w_i<^{(n)} f$.
			\end{enumerate}  
			\begin{figure}[htb!]
				\includegraphics[width = 85mm,scale = 1.5]{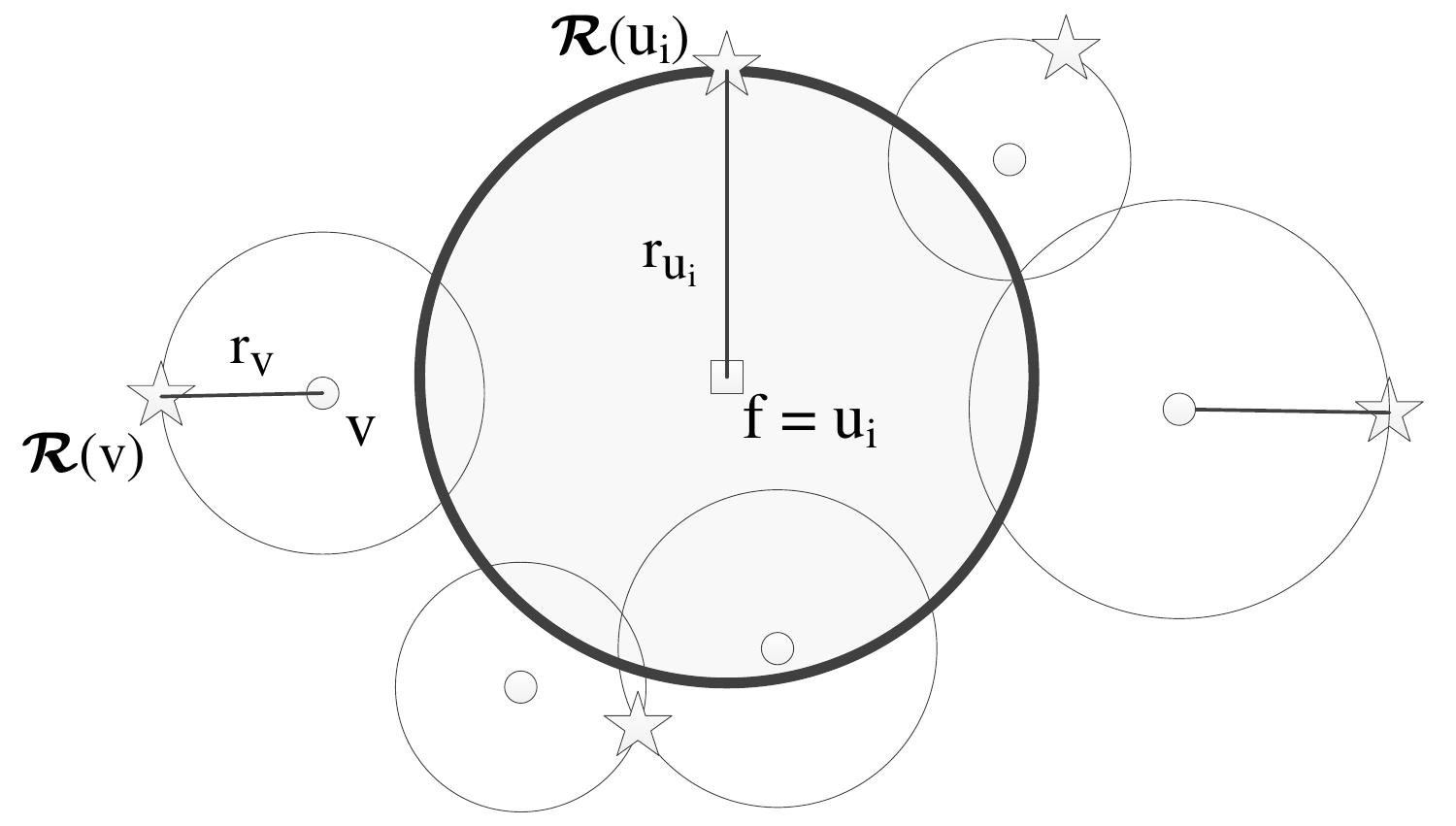}
				\centering
				\caption{Disk $D_{u_i}=D_f$ (with bold border), and the disks in $\mathcal{D}$ that are not bigger than $D_{u_i}$ and intersect with $D_{u_i}$. 
					Asterisks represent the nodes in $\mathcal{S}_i$.}
				\label{fig:Df-remain}
			\end{figure}
			To sum up, there is a node $w_i <^{(n)} f$ that has received the message, and transmits with a power at least equal to $(5r_{u_i})^\alpha$, 
			and is at most $3r_{u_i}$ away from $f$.
			Consequently $f$ must have received the message from $w_i$.
		
		\item Case 2: $D_f \not\in \mathcal{I}$:\\
			Since $D_f\notin \mathcal{I}$, there must be at least a disk in $\mathcal{I}$ that intersects with $D_f$. Among such disks, let $D_{u_i}$ be the largest one. 
			The disk $D_{u_i}$ must be at least as large as $D_f$ as otherwise, $D_f$ would have entered the set $\mathcal{I}$ instead of $D_{u_i}$. 
			In Fig.~\ref{fig:Df-discard}, the disk $D_{u_i}$ is shown with thick boundary. 
			The Figure also shows all the disks in $\mathcal{D}$ that intersect with $D_{u_i}$ and are at most as large as $D_{u_i}$.
			Similar to Case 1, every node in $\mathcal{S}_i$ is represented by an asterisk, and $w_i$ is one of the asterisks. 
			Similar to Case 1, we have
			\begin{enumerate} [label=(\roman*)]
				\item $\rho^{(n)}(w_i) \geq (5r_{u_i})^\alpha$.
				\item $d_{w_i,u_i} \leq 5r_{u_i}$:\\
				This is because the radius of all disks is at most $r_{u_i}$. 
				As illustrated in Fig.~\ref{fig:Df-discard}, any point on the circumference of any disk is at most $5r_{u_i}$ away from node $f$.
				\item $w_i <^{(n)} f$:\\
				From (\ref{equ:resp}), we get $\mathcal{R}(f) <^{{(c)}} f$, and (\ref{equ:w_i}) implies that $w_i <^{(c)} \mathcal{R}(f)$. 
				Hence, by (\ref{equ:order-m}), we get that $w_i <^{(n)} f$.
				\item $w_i$ has received the message:\\
				Again, this is by the assumption that every node $u<^{(n)} f$ has received the message, and because $w_i <^{(n)} f$.
			\end{enumerate} 
			\begin{figure}[htb!]
				\includegraphics[width = 85mm,scale = 1]{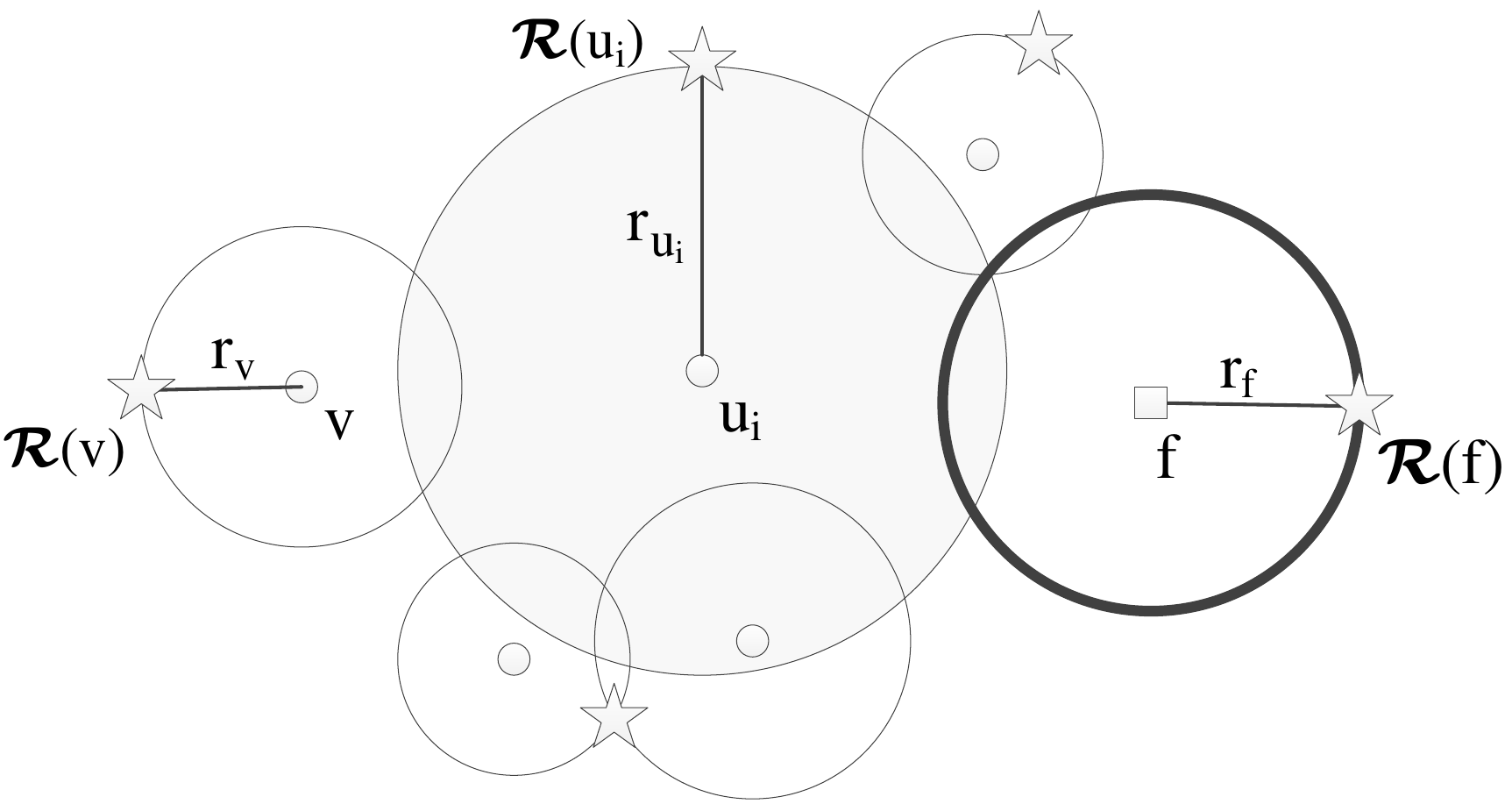}
				\centering
				\caption{The shaded disk $D_{u_i}$ is the biggest disk in $\mathcal{I}$ that intersects with $D_f$ (the disk with bold border).}
				\label{fig:Df-discard}
			\end{figure}
			To sum up, node $w_i$ has received the message, $w_i<^{(n)} f$. 
			It transmits with a power at least equal to $(5r_{u_i})^\alpha$, and is at most $5r_{u_i}$ away from $f$. 
			Consequently, as in Case 1, $f$ must have received the message from $w_i$.
	\end{enumerate}

\section{Proof of Proposition~\ref{prp:bri}}
\label{app:bri}
  
  Let $c'_i$ denote the center of $D'_{u_i}$.
  There is a node located at $c'_i$, and that node can decode the message when the algorithm $(\mathscr{A}^{(c)},\rho^{(c)})$ is used.
  Therefore, the sum of received power  from nodes in $\mathscr{T}_{u_i}$ at $c'_i$ must be at least equal to the decoding threshold.
  Let $v$ be any node in $\mathscr{T}_{u_i}$ and $p'$ be any point in $D'_{u_i}$.
  We have
  \[
    d_{v,p'} \leq d_{v,c'_i}+r'_i/\gamma,
  \]
  where $d_{v,p'}$ denotes the distance between $v$ and $p'$ (similarly for  $d_{v,c'_i}$), and $r'_i$ denotes the radius of $D'_{u_i}$.
  We also have
  \[
    d_{v,c'}\geq r'_i
  \]
  as $v$ is in $\mathscr{T}_{u_i}$, hence cannot be inside $D'_{u_i}$.
  Therefore, the received power from $v$ at point $p'$ is at least
  \[
  \begin{split}
    \left(\frac{d_{v,c'_i}}{d_{v,c'_i}+r'_i/\gamma} \right)^\alpha
    &= \left(\frac{1}{1+\frac{r'_i/\gamma}{d_{v,c'}}} \right)^\alpha \\
    &\geq \left(\frac{1}{1+\frac{1}{\gamma}} \right)^\alpha
  \end{split}
  \]
  fraction of the received power from $v$ at the center.
  Therefore, if we multiply all the transmission powers (which includes every transmission power in $\mathscr{T}_{u_i}$)
  by a factor $(1+\frac{1}{\gamma})^\alpha$ (which we do in the first step of the roadmap), 
  the total received powers from $\mathscr{T}_{u_i}$ at $p'$ will be at least equal to the decoding threshold.

\section{Proof of Theorem~\ref{thm:valid}}
\label{app:valid}
  The proof is trivial for $n=1$, thus we assume that $n\geq 2$.
  Without loss of generality,  assume that $D_n$ is the largest disk, that is $r_n\geq r_i$, for every $1\leq i\leq n$. 
  
  For every $1\leq i \leq n$, let $u_i(p): \mathbb{R}^2 \rightarrow \mathbb{R}$ be a function such that
    \[
    u_i(p)=
      \begin{cases}
        r_i^{\alpha-2}		& \text{if } \dist (p, c_i)\leq r_i;\\
        0						& \text{otherwise}.
      \end{cases}
    \]  
  In proving the theorem, we use double integrations in several places, including the following lemma.
  In all those places, we assume a polar coordinate system with the pole at  $c_n$.
  We start by proving some lemmas.  
%
%
%

  \begin{lemma}
  \label{lem:ui}
     For every $1\leq i \leq n-1$, we have
     \[
       \int\int_{D_i}  \frac{u_i(r,\theta)}{r^{\alpha}}~r\diff r \diff \theta \geq 2\pi\left(\frac{2\gamma-1}{2\gamma+1}\right)^\alpha \left( \frac{r_i}{\dist(c_n, c_i)-r_i} \right)^\alpha
     \]      
  \end{lemma}
  \begin{proof}
    For every point $p$ , $p \in D_i$, we have 
    \[
      \dist(c_n, c_i)-r_i \leq \dist(c_n, p) \leq  \dist(c_n, c_i) +r_i.
    \]
    Therefore, 
    \[
      \int\int_{D_i}  \frac{1}{r^{\alpha}}~r \diff r \diff \theta \geq \frac{2\pi r_i^2}{ (\dist(c_n, c_i) + r_i)^\alpha},
    \]
    because 
    \[
        \int\int_{D_i} r \diff r \diff \theta = 2\pi r_i^2,
    \]
    and
    \[
       \frac{1}{r^{\alpha}} \geq \frac{1}{(\dist(c_n, c_i) +r_i)^\alpha}
    \]
    in the region $D_i$.
    Thus 
    \begin{equation}
    \label{equ:dist}
    \begin{split}
       \int\int_{D_i}  \frac{u_i(r,\theta)}{r^{\alpha}}~r \diff r \diff \theta 
       &= \int\int_{D_i}  \frac{r_i^{\alpha-2}}{r^{\alpha}}~r \diff r \diff \theta \\
       &\geq r_i^{\alpha-2} \cdot  \frac{2\pi r_i^2}{ (\dist(c_n, c_i) + r_i)^\alpha}\\
       &=\frac{2\pi r_i^\alpha}{ (\dist(c_n, c_i) + r_i)^\alpha}\\
     \end{split}
    \end{equation}

\noindent
    We have $r_n\geq r_i$, and
    \[
      \dist(c_n, c_i) \geq \gamma(r_n+r_i) \geq 2\gamma r_i,
    \]
   With some simple calculation, we get
   \[
     \frac{1}{\dist(c_n, c_i)+r_i} \geq \left( \frac{2\gamma-1}{2\gamma+1} \right)\cdot \frac{1}{\dist(c_n, c_i)-r_i}
   \]
   Therefore, by (\ref{equ:dist}), we get
   \[
     \int\int_{D_i}  \frac{u_i(r,\theta)}{r^{\alpha}}~r\diff r \diff \theta \geq 2\pi\left(\frac{2\gamma-1}{2\gamma+1}\right)^\alpha \left( \frac{r_i}{\dist(c_n, c_i)-r_i} \right)^\alpha
   \]

  \end{proof}

  \begin{lemma}
  \label{lem:sumRi}
    We have
    \[
      \sum_{i=1}^{n-1} \left( \frac{r_i}{\dist(c_n, c_i)-r_i} \right)^\alpha \leq  \beta_1,
    \]
    where
    \[
      \beta_1=\frac{1}{(\alpha-2) \cdot \gamma^{\alpha-2}}\cdot\left(\frac{2\gamma+1}{2\gamma-1}\right)^\alpha.
    \]
  \end{lemma}
  \begin{proof}
    By Lemma~\ref{lem:ui}, we get
    \[
    \begin{split}
      &\sum_{i=1}^{n-1} \left( \frac{r_i}{\dist(c_n, c_i)-r_i} \right)^\alpha\\
      &\leq  \frac{1}{2\pi}\cdot\left(\frac{2\gamma+1}{2\gamma-1}\right)^\alpha  \left( \sum_{i=1}^{n-1} \int\int_{D_i}  \frac{u_i(r,\theta)}{r^{\alpha}}~r\diff r \diff \theta \right) \\
      &=  \frac{1}{2\pi}\cdot\left(\frac{2\gamma+1}{2\gamma-1}\right)^\alpha  \left( \int\int_{\cup_{i=1}^{n-1} D_i}  \frac{u_i(r,\theta)}{r^{\alpha}}~r\diff r \diff \theta \right)      \\
      &\leq  \frac{1}{2\pi}\cdot\left(\frac{2\gamma+1}{2\gamma-1}\right)^\alpha  \left( \int_{0}^{2\pi}\int_{\gamma r_n}^{\infty}  \frac{u_i(r,\theta)}{r^{\alpha}}~r\diff r \diff \theta \right)      \\      
      &\leq  \frac{1}{2\pi}\cdot\left(\frac{2\gamma+1}{2\gamma-1}\right)^\alpha  \left( \int_{0}^{2\pi}\int_{\gamma r_n}^{\infty}  \frac{r_n^{\alpha-2}}{r^{\alpha}}~r\diff r \diff \theta \right)      \\     
      &\leq  \frac{r_n^{\alpha-2}}{2\pi}\cdot\left(\frac{2\gamma+1}{2\gamma-1}\right)^\alpha  \left( \int_{0}^{2\pi}\int_{\gamma r_n}^{\infty}  \frac{1}{r^{\alpha}}~r\diff r \diff \theta \right)      \\  
      &=  \frac{r_n^{\alpha-2}}{2\pi}\cdot\left(\frac{2\gamma+1}{2\gamma-1}\right)^\alpha  \left( 2\pi \cdot \frac{r^{2-\alpha}}{2-\alpha} \left|_{\gamma r_n}^{\infty} \right. \right)      \\                        
      &=  \frac{r_n^{\alpha-2}}{2\pi}\cdot\left(\frac{2\gamma+1}{2\gamma-1}\right)^\alpha  \left( 2\pi \cdot \frac{(\gamma r_n)^{2-\alpha}}{\alpha-2}  \right)      \\     
      &=  \frac{1}{(\alpha-2)\cdot \gamma^{\alpha-2}}\cdot\left(\frac{2\gamma+1}{2\gamma-1}\right)^\alpha = \beta_1.
    \end{split}
    \]
  \end{proof}  
  
  The following lemma is similar to Lemma~\ref{lem:sumRi}, except, $r_i$ is replaced with $r_n$ in the denominator. 
  \begin{lemma}
  \label{lem:sumRn}  
    We have
    \[
      \sum_{i=1}^{n-1} \left( \frac{r_i}{\dist(c_n, c_i)-r_n} \right)^\alpha \leq \beta_2 
    \]
    where
    \[
    \begin{split}    
       \beta_2
       &=  \beta_1\left( \frac{\gamma}{\gamma-1} \right)^\alpha \\
       &= \frac{1}{(\alpha-2) \cdot \gamma^{\alpha-2}}\cdot\left(\frac{\gamma \cdot (2\gamma+1)}{(\gamma-1) \cdot (2\gamma-1)}\right)^\alpha.
    \end{split}       
    \]
    \end{lemma}  
    \begin{proof}
     We have $r_n \geq r_i$, and $\dist(r_n, r_i)\geq \gamma(r_n+r_i)$.
     With a simple calculation, we get
      \[
        \frac{\dist(r_n, r_i)-r_i}{\dist(r_n, r_i)-r_n} \leq \frac{\gamma}{\gamma-1}.
      \]
     Combining this inequality with the inequality stated in Lemma~\ref{lem:sumRi}, we get the result.
  \end{proof}
  
  \begin{lemma}
  \label{lem:maxPn}
    Let $p_1, \ldots, p_n$ be a valid power assignment in the brightening non-overlapping disks problem.
    Assume that $D_n$ is the largest disks.
    Then,
    \[
      p_n \geq (1-\beta_2) \cdot r_n^\alpha ,
    \]
    where
    \[
      \beta_2 = 1- \frac{1}{(\alpha-2) \cdot \gamma^{\alpha-2}}\cdot\left(\frac{\gamma \cdot (2\gamma+1)}{(\gamma-1) \cdot (2\gamma-1)}\right)^\alpha
    \]
  \end{lemma}
  \begin{proof}
    Let $q$ be a point on the circumference of $D_n$.
    The received power from $c_i$, $1\leq i\leq n-1$ at point $q$ is
    \[
    \begin{split}
      \frac{p_i}{\dist(c_i, q)^\alpha}
      &\leq  \frac{p_i}{(\dist(c_i, c_n)-r_n)^\alpha}\\
      &\leq  \frac{r_i^{\alpha}}{(\dist(c_i, c_n)-r_n)^\alpha}\\      
    \end{split}
    \]
    Then, the total power received at point $q$ from $c_1, c_2, \ldots, c_{n-1}$ is
    \[
    \begin{split}
      &\sum_{i=1}^{n-1} \frac{p_i}{\dist(c_i, q)^\alpha}\\
      &\leq \sum_{i=1}^{n-1}  \frac{r_i^{\alpha}}{(\dist(c_i, c_n)-r_n)^\alpha}\\     
      &\leq  \frac{1}{(\alpha-2) \cdot \gamma^{\alpha-2}}\cdot\left(\frac{\gamma \cdot (2\gamma+1)}{(\gamma-1) \cdot (2\gamma-1)}\right)^\alpha \\
      &= \beta_2
    \end{split}
    \]    
    where the last inequality is by Lemma~\ref{lem:sumRn}.
    Since the total received power at $q$ must be at least equal to one, we get that the received power at $q$ from $c_n$ must be at least,
    \[
    \begin{split}    
      1- \sum_{i=1}^{n-1} \frac{p_i}{\dist(c_i, q)^\alpha} 
      &\geq 1- \beta_2
    \end{split}
    \]
    hence
    \[
      p_n \geq (1- \beta_2) \cdot r_n^\alpha.
    \]
  \end{proof}

  \begin{lemma}
  \label{lem:validAssign}
    Let $p_1, \ldots, p_n$ be a valid power assignment in the brightening non-overlapping disks problem.
    Then, a valid power assignment for the brightening non-overlapping disks problem
    with the set of disks $D_1, D_2, \ldots, D_n$ (i.e., all the disks except the largest disk $D_n$) is
    \[
      p_i+p_n  \left( \frac{r_i}{\dist(c_n, c_i)-r_i} \right)^\alpha  \qquad 1\leq i \leq n-1.
    \]
  \end{lemma}
  
  The proof of theorem is by induction on the number of disks~$n$.
  
  Let $p_i$, $1 \leq i \leq n$ be a solution to the brightening non-overlapping disk problem with disks $D_1, \ldots, D_n$.
  That is, $p_i$, $1 \leq i \leq n$ is a valid power assignment with minimum total sum of powers among all valid power assignments.
  Also, let  $p'_i$, $1 \leq i \leq n$ be a solution to the brightening non-overlapping disk problem with disks $D_1, \ldots, D_{n-1}$ (i.e., all the disks except the largest disk $D_n$).
  By Lemma~\ref{lem:validAssign}, we know that
    \[
      p_i+p_n  \left( \frac{r_i}{\dist(c_n, c_i)-r_i} \right)^\alpha  \qquad 1\leq i \leq n-1.
    \]  
  is a valid power assignment  to the brightening non-overlapping disk problem with disks $D_1, \ldots, D_{n-1}$.
  Therefore, we must have
  \begin{equation}
  \label{equ:p'Bound}
    \sum_{i=1}^{n-1} p'_i \leq \sum_{i-1}^{n-1} \left( p_i+p_n \cdot \left( \frac{r_i}{\dist(c_n, c_i)-r_i} \right)^\alpha  \right).
  \end{equation}
  because $p'_i$, $1\leq i\leq n-1$, is a valid power assignment with minimum total sum power.
  By the induction hypothesis
  \begin{equation}
  \label{equ:hyp}
    \sum_{i=1}^{n-1} p'_i \geq \theta \cdot \sum_{i=1}^{n-1} r_i^\alpha.
  \end{equation}
  Therefore, by (\ref{equ:p'Bound}) and (\ref{equ:hyp}), we get
  \[
  \begin{split}
    \sum_{i=1}^{n-1} \left( p_i+p_n \cdot \left( \frac{r_i}{\dist(c_n, c_i)-r_i} \right)^\alpha  \right) 
    &\geq  \sum_{i=1}^{n-1} p'_i \\
    & \geq \theta \cdot \sum_{i=1}^{n-1} r_i^\alpha.
  \end{split}
  \]
  Thus
 \[
 \begin{split}
    \sum_{i=1}^{n-1} p_i  
    &\geq  \sum_{i=1}^{n-1} p'_i -  \sum_{i=1}^{n-1} p_n \cdot \left( \frac{r_i}{\dist(c_n, c_i)-r_i} \right)^\alpha \\
    &\geq   \theta \cdot \sum_{i=1}^{n-1} r_i^\alpha -  \sum_{i=1}^{n-1} p_n \cdot \left( \frac{r_i}{\dist(c_n, c_i)-r_i} \right)^\alpha \\    
    &=   \theta \cdot \sum_{i=1}^{n-1} r_i^\alpha -  p_n \cdot \sum_{i=1}^{n-1}  \left( \frac{r_i}{\dist(c_n, c_i)-r_i} \right)^\alpha \\     
    &\geq   \theta \cdot \sum_{i=1}^{n-1} r_i^\alpha - \beta_1 \cdot p_n \\     
 \end{split}
  \]  
  where the last inequality is by Lemma~\ref{lem:sumRi}.
  Consequently, by Lemma~\ref{lem:maxPn}, we get
  
 \[
 \begin{split}
    \sum_{i=1}^{n} p_i 
    &=\sum_{i=1}^{n-1} p_i + p_n\\
   &\geq   \beta \cdot \sum_{i=1}^{n-1} r_i^\alpha - \beta_1 \cdot p_n  + (1-\beta_2) \cdot p_n\\     
   &=\beta \cdot \sum_{i=1}^{n-1} r_i^\alpha + (1-\beta_1- \beta_2) \cdot p_n \\     
   &=\beta \cdot \sum_{i=1}^{n} r_i^\alpha, \\        
 \end{split}
  \]    
  where the last equality is because
  \[
    \beta =1-\beta_1-\beta_2.
  \]

\section{Proof of Theorem~\ref{the:consRatio}}
\label{app:consRatio}
We start by some definitions and lemmas.
Let $(\mathscr{A}^{(c)},\rho^{(c)})$ be a cooperative broadcast algorithm, and $(\mathscr{A}^{(n)},\rho^{(n)})$ be the non-cooperative broadcast algorithm constructed from it.
\begin{definition}
	\label{def:f_u(p)}
	For every node $u \in \mathscr{A}^{(c)}$, we define
	\begin{equation}
	\begin{split}
	\label{equ:mod-func}
	&f_u: \mathbb{R}^2 \to (0,1]\\
	&f_u(p) = 
	\left\{
	\begin{array}{ll}
	1 & d_{u,p} \leq \sqrt{\rho^{(c)}(u)} \\
	\frac{\rho^{(c)}(u)}{d_{u,p}^2} & d_{u,p} > \sqrt{\rho^{(c)}(u)}    
	\end{array}  
	\right.
	\end{split}
	\end{equation} 
	where $p$ is a point in the network, and $d_{u,p}$ is the distance between $p$ and node $u$.
	The function $f_u(p)$ is simply the received power from node $u$ at point $p$ capped at one.
	Function~$\mathcal{F}$, accordingly, is defined as 
	\begin{equation*}
	\begin{split}
	\label{equ:agg-func}
	&\mathcal{F}: \mathbb{R}^2 \to \mathbb{R}\\
	&\mathcal{F}(p) = \sum_ {u\in\mathscr{U}} f_u(p).
	\end{split}
	\end{equation*}
\end{definition}
\begin{lemma}
	\label{equ:entropy}
	Let $|\mathscr{U}|$ = n. We have
	\begin{equation*}
	\sum_{u \in \mathscr{A}^{(c)}} \frac{\rho^{(c)}(u)}{P_{tot}^{(c)}}  \ln \left(\frac{ P_{tot}^{(c)}}{\rho^{(c)}(u)}\right) \leq \ln(n)
	\end{equation*}
\end{lemma}

\begin{proof}
	Let $\alpha_u = \frac{\rho^{(c)}(u)}{P_{tot}^{(c)}}$.
	We have $\sum_{u \in \mathscr{A}^{(c)}}^{} \alpha_u = 1$.
	Therefore, the sequence $\alpha_u$ can be seen as a probability distribution whose entropy is
	\begin{equation*}
	\sum_{u \in \mathscr{A}^{(c)}}^{} \alpha_u \ln (\frac{1}{\alpha_u}) \leq \ln (|\mathscr{A}^{(c)}|) \leq \ln (|\mathscr{U}|) = \ln (n)
	\end{equation*}	
\end{proof}
\begin{lemma}
	\label{lem:4}
	For any real number $R$, $R \geq \sqrt{\rho^{(c)}(u)}$, we have
	\begin{equation*}
	\int_{0}^{2\pi}\int_{0}^{R} f_u(r,\theta)r \diff r\diff \theta = \pi \rho^{(c)}(u) + \pi \rho^{(c)}(u) \ln \left(\frac{R^2}{\rho^{(c)}(u)}\right),
	\end{equation*} 
	where the function $f_u(r,\theta)$ is the function defined in (\ref{equ:mod-func}) transferred into the polar coordinate system with the pole at node $u$.
\end{lemma}

\begin{proof}
	\begin{equation*}
	\begin{split}
	\label{equ-agg-int}
	\int_{0}^{2\pi} \int_{0}^{R} f_{u}(r,\theta)r \diff r\diff \theta =&
	\int_{0}^{2\pi} \int_{0}^{\sqrt{\rho^{(c)}(u)}} f_{u}(r,\theta)r \diff r\diff \theta ~
	+\\ &\int_{0}^{2\pi} \int_{\sqrt{\rho^{(c)}(u)}}^{R} f_{u}(r,\theta)r \diff r\diff \theta\\
	=&\int_{0}^{2\pi} \int_{0}^{\sqrt{\rho^{(c)}(u)}} r \diff r\diff \theta ~ 
	+\\ &\int_{0}^{2\pi} \int_{\sqrt{\rho^{(c)}(u)}}^{R} \frac{\rho^{(c)}(u)}{r^2}r \diff r\diff \theta\\
	=&  \pi \rho^{(c)}(u) + \pi \rho^{(c)}(u) \ln \left(\frac{R^2}{\rho^{(c)}(u)}\right).\\
	\end{split}
	\end{equation*}
\end{proof}
\noindent Let $U_\mathcal{I} = \bigcup\limits_{D \in \mathcal{I}} D$, and $\Delta_\mathcal{I}$ be the area of~$U_\mathcal{I}$.
\begin{lemma}
	We have
	\begin{equation*}
	\label{lem:5}
	\iint_{U_\mathcal{I}}f_{u}(r,\theta)r \diff r\diff \theta \leq \pi \rho^{(c)}(u) + \pi \rho^{(c)}(u) \ln (\frac{\Delta_\mathcal{I}}{\pi\rho^{(c)}(u)})
	\end{equation*}
\end{lemma}
\begin{proof}
	Let $D_\mathcal{I}$ be a disk with radius $R = \sqrt{\frac{\Delta_\mathcal{I}}{\pi}}$ centered at $u$. Let $U_{in}=U_\mathcal{I}\cap D_\mathcal{I}$, and $U_{out}=U_\mathcal{I}\backslash D_\mathcal{I}$, be the parts of $U_\mathcal{I}$ that respectively fall inside and outside of the disk $D_\mathcal{I}$.
	Let $q$ be a point on the circumference of disk $D_\mathcal{I}$. For any point $p_{in}\in D_\mathcal{I}$ we have 
	\[
	f_u(p_{in})\geq f_u(q), 
	\]  
	since $d_{u, p_{in}} \leq d_{u,q}$, and the function $f_u(p)$ is non-increasing in terms of $d_{u,p}$.
	Similarly, we have
	\[
	f_u(p_{out})\leq f_u(q),
	\] 
	for any point $p_{out} \notin D_I$.
	Therefore, we have  
	\begin{equation}
	\begin{split}
	\label{equ:lem-dual-5}
	\iint_{U_{out}}f_{u}(r,\theta)r \diff r\diff \theta &\leq \Delta_{U_{out}} \times f_{u}(q)
	\Delta_{D_\mathcal{I} \backslash U_{in}} \times f_{u}(q) \\&\leq \iint_{D_\mathcal{I} \backslash U_{in}} f_{u}(r,\theta)r \diff r\diff \theta,
	\end{split}
	\end{equation}
	where $\Delta_{U_{out}}$ and $\Delta_{D_\mathcal{I} \backslash U_{in}}$ are areas of $U_{out}$ and $D_\mathcal{I} \backslash U_{in}$, respectively.
	Note that $\Delta_{U_{out}} = \Delta_{D_\mathcal{I} \backslash U_{in}}$. 
	Consequently,
	\[
	\begin{split}
	&\iint_{U_\mathcal{I}}f_u(r,\theta)r \diff r\diff \theta
	=\\&\iint_{U_{in}}f_u(r,\theta)r \diff r\diff \theta +\iint_{U_{out}}f_u(r,\theta)r \diff r\diff \theta
	\leq\\&
	\int_{U_{in}}f_u(r,\theta)r \diff r\diff \theta +\iint_{U_\mathcal{I}\backslash U_{in}}f_u(r,\theta)r \diff r\diff \theta
	=\\&\iint_{D_\mathcal{I}}f_u(r,\theta)r \diff r\diff \theta
	\leq \pi \rho^{(c)}(u) + \pi \rho^{(c)}(u) \ln (\frac{\Delta_\mathcal{I}}{\pi\rho^{(c)}(u)})
	\end{split}
	\]
	where the last inequality is by Lemma 5.
	
\end{proof}
\begin{lemma}
	\label{lemma:6}
	For any point $p \in U_\mathcal{I}$, we have 
	\[
	\mathcal{F}(p) \geq \frac{1}{4}
	\]
\end{lemma}

\begin{proof}
	Let $p$ be an arbitrary point in $U_\mathcal{I}$.
	Then, $p$ must be inside one disk in $\mathcal{I}$. 
	Let that disk be $D_v\in \mathcal{I}$.
	Since in $(\mathscr{A}^{(c)}, \rho^{(c)})$ node $v$ receives the message, we must have
	\begin{equation}
	\label{equ:F(v)}
	\sum_{u<^{(c)}v} f_u(v)\geq 1.
	\end{equation}
	For any node $u<^{(c)}v$, we have $d_{u,v} \geq r_v$,
	where $r_v$ is the radius of disk $D_v$.
	Therefore, for any node $u<^{(c)}v$, we have
	\begin{equation}
	\label{equ:dist}
	d_{u,p} \leq 2 d_{u,v}.
	\end{equation}
	
	Hence, by (\ref{equ:F(v)}) and (\ref{equ:dist}), we get
	\begin{equation*}
	\begin{split}
	\mathcal{F}(p) \geq \sum_{u<^{(c)}v}^{} f_{u}(p) &=
	\sum_{u<^{(c)}v}^{} \left(f_{u}(v) \times \left(\frac{d_{u,v}}{d_{u,p}}\right)^2\right) 
	\\&\geq \sum_{u<^{(c)}v}^{} \left(f_{u}(v) \times \frac{1}{4}\right) \geq
	\frac{1}{4}.
	\end{split}
	\end{equation*}
\end{proof}

The following corollary directly follows from Lemma \ref{lemma:6}.
\begin{cor}
	\label{cor:7}
	we have
	\[
	\iint_{U_\mathcal{I}} \mathcal{F}(p) \geq \frac{1}{4} \Delta_\mathcal{I}.
	\] 
\end{cor}

Let $r_{u_i}$ denote the radius of disk $D_{u_i} \in \mathcal{I}$, \mbox{$1 \leq i \leq |\mathcal{I}|$}.
We have
\begin{equation*}
\label{cor:8-1}
\Delta_\mathcal{I} = \pi \sum_{i = 1}^{|\mathcal{I}|} r_{u_i}^2.
\end{equation*}
In Algorithm \ref{alg:1}, each disk $D_{u_i}$ , $1 \leq i \leq |\mathcal{I}|$, contributes a power of $25r_{u_i}^2$ at most once. Therefore, 
\begin{equation*}
\label{cor:8-2}
P_{tot}^{(n)} \leq 25 \sum_{i=1}^{|\mathcal{I}|} r_{u_i}^2
\end{equation*}
Hence, 
\begin{equation}
\label{equ:delta-I}
\Delta_\mathcal{I} \geq \frac{\pi}{25}P_{tot}^{(n)}
\end{equation}

  	By Corollary \ref{cor:7}, we have
  	\[
  	\iint_{U_\mathcal{I}} \mathcal{F}(p) \geq \frac{1}{4} \Delta_\mathcal{I}.
  	\]
  	To prove the theorem, we compute an upper bound on the same integration, and then compare the two bounds.
  	%
  	\begin{equation}
  	\begin{split}
  	\label{prop:1}
  	\iint_{U_\mathcal{I}} \mathcal{F}(p) = 
  	&\iint_{U_\mathcal{I}} \sum_ {u \in \mathscr{A}^{(c)}} f_u(p) \\
  	=&\sum_{u \in \mathscr{A}^{(c)}} \iint_ {U_\mathcal{I}} f_u(p)\\
  	\leq& \sum_{u \in \mathscr{A}^{(c)}} \pi \rho^{(c)}(u) +  \sum_{u \in \mathscr{A}^{(c)}} \pi \rho^{(c)}(u) \ln (\frac{\Delta_\mathcal{I}}{\pi\rho^{(c)}(u)})
  	\end{split}
  	\end{equation} 
  	The second summation can be written as
  	\begin{equation*}
  	\begin{split}
  	&\sum_{u \in \mathscr{A}^{(c)}} \pi \rho^{(c)}(u) \ln (\frac{\Delta_\mathcal{I}}{\pi\rho^{(c)}(u)})=\\
	& \pi P_{tot}^{(c)} \sum_{u \in \mathscr{A}^{(c)}} \frac{\rho^{(c)}(u)}{P_{tot}^{(c)}}  \ln \left(\frac{\pi P_{tot}^{(c)}}{\pi\rho^{(c)}(u)} \times \frac{\Delta_\mathcal{I}}{\pi P_{tot}^{(c)}}\right)=\\
	&  \pi P_{tot}^{(c)} \sum_{u \in \mathscr{A}^{(c)}} \frac{\rho^{(c)}(u)}{P_{tot}^{(c)}}  \ln \left(\frac{ P_{tot}^{(c)}}{\rho^{(c)}(u)}\right)+ \\
  	&\sum_{u \in \mathscr{A}^{(c)}} \pi \rho^{(c)}(u) \ln \left(\frac{\Delta_\mathcal{I}}{\pi P_{tot}^{(c)}}\right)\\
  	&\leq \pi P_{tot}^{(c)} \sum_{u \in \mathscr{A}^{(c)}} \frac{\rho^{(c)}(u)}{P_{tot}^{(c)}}  \ln \left(\frac{ P_{tot}^{(c)}}{\rho^{(c)}(u)}\right) 
  	+	\pi P_{tot}^{(c)} \ln \left(\frac{P_{tot}^{(n)}}{25 P_{tot}^{(c)}}\right).
  	\end{split}
  	\end{equation*}
  	where the last inequality is by (\ref{equ:delta-I}). 
  	Furthermore, by Lemma~\ref{equ:entropy} we get
  	\begin{equation}
  	\begin{split}
  	\label{prop:2}
  	&\pi P_{tot}^{(c)} \sum_{u \in \mathscr{A}^{(c)}} \frac{\rho^{(c)}(u)}{P_{tot}^{(c)}}  \ln \left(\frac{ P_{tot}^{(c)}}{\rho^{(c)}(u)}\right) 
  	+	\pi P_{tot}^{(c)} \ln \left(\frac{P_{tot}^{(n)}}{25 P_{tot}^{(c)}}\right)\\  
  	&\leq \pi P_{tot}^{(c)} \ln(n) +	\pi P_{tot}^{(c)} \ln \left(\frac{P_{tot}^{(n)}}{25 P_{tot}^{(c)}}\right).
  	\end{split}
  	\end{equation}	
  	Using (\ref{prop:2}) in (\ref{prop:1}), we have
  	\begin{equation}
  	\begin{split}
  	\label{equ:F_T}
  	\iint_{U_\mathcal{I}} \mathcal{F}(p) 
  	&\leq \sum_{u \in \mathscr{A}^{(c)}} \pi \rho^{(c)}(u) +  \sum_{u \in \mathscr{A}^{(c)}} \pi \rho^{(c)}(u) \ln (\frac{\Delta_\mathcal{I}}{\pi\rho^{(c)}(u)})\\
  	&\leq \pi P_{tot}^{(c)} + \pi P_{tot}^{(c)} \ln(n) +	\pi P_{tot}^{(c)} \ln \left(\frac{P_{tot}^{(n)}}{25 P_{tot}^{(c)}}\right).
  	\end{split}
  	\end{equation}
  	By corollary \ref{cor:7} and (\ref{equ:delta-I}), we get
  	\begin{equation}
  	\label{equ:cor-7+8}
  	\iint_{U_\mathcal{I}} \mathcal{F}(p) \geq \frac{1}{4} \Delta_\mathcal{I} \geq\frac{\pi}{100}P_{tot}^{(n)}. 
  	\end{equation} 
  	From (\ref{equ:F_T}) and (\ref{equ:cor-7+8}), we get
  	\begin{equation}
  	\label{equ:before-Gtot}
  	\pi P_{tot}^{(c)} + \pi P_{tot}^{(c)} \ln(n) +	\pi P_{tot}^{(c)} \ln \left(\frac{P_{tot}^{(n)}}{25 P_{tot}^{(c)}}\right) \geq \frac{\pi}{100}P_{tot}^{(n)}
  	\end{equation}	
  	Dividing both sides by $\pi P_{tot}^{(c)}$ yields
  	\begin{equation}
  	\label{equ:Gtot}
  	1 + \ln(n) + \ln(\frac{G_{tot}}{25}) \geq \frac{G_{tot}}{100}
  	\end{equation}
  	Finally, assuming $n\geq 2$, it can be verified that (\ref{equ:Gtot}) holds only when $G_{tot} \leq 127 \ln(n)$, which completes the proof.

\section{Proof of Proposition~\ref{prp:grd1}}
\label{app:grd1}
	Let $(\mathscr{A}^{(n)},\rho^{(n)})$ be an arbitrary non-cooperative broadcast algorithm.
	If there is a node $u \in \mathscr{A}^{(n)}$ with $\rho^{(n)}(u) < d^2$, we can safely set $\rho^{(n)}(u)=0$ (hence removing $u$ from $\mathscr{A}^{(n)}$), since the transmission by $u$ will not reach any other node in the grid network.
	For any node $u \in \mathscr{A}^{(n)}$, let $\mathcal{N}_u$ denote the set of nodes within the transmission range of $u$, that is
	\[
	\mathcal{N}_u = \{v|d_{u,v}\leq\sqrt{\rho(u)} \}, 
	\]
	where $d_{u,v}$ denotes the distance between nodes $u,v \in \mathscr{U}$.
	Note that all nodes in $\mathcal{N}_u$ are within the square with side length of 
	$2\sqrt{\rho(u)}$ entered at $u$, because they are all inside the circle with radius \mbox{$\sqrt{\rho(u)}$} centered at $u$.
	The number of nodes within a square with sides length of $2\sqrt{\rho(u)}$ is bounded by
	\[
	|\mathcal{N}_u| \leq \left(2\left\lfloor\frac{\sqrt{\rho(u)}}{d}\right\rfloor+1\right)^2.
	\label{equ:Cu}
	\]
	Therefore
	\begin{equation}
	|\mathcal{N}_u| \leq 9\left(\frac{\rho(u)}{d^2}\right),
	\label{equ:|Cu|}
	\end{equation}
	because
	\[
	\begin{split}
	\left(2\left\lfloor\frac{\sqrt{\rho(u)}}{d}\right\rfloor+1\right)^2 
	&\leq \left(2\frac{\sqrt{\rho(u)}}{d}+1\right)^2  \\
	&\leq \left(3\frac{\sqrt{\rho(u)}}{d}\right)^2=9\left(\frac{\rho(u)}{d^2}\right),
	\end{split}
	\]
	where the last inequality holds since $\rho(u) \geq d^2$ (i.e. $\frac{\sqrt{\rho(u)}}{d} \geq 1$).
	Every node in the network, including the source, is within the transmission range of at least one transmitting node in $\mathscr{A}^{(n)}$.
	Thus, we must have
	\[
	\sum_{\forall u\in \mathscr{A}^{(n)}}^{} |\mathcal{N}_u| \geq n.
	\] 
	Hence by (\ref{equ:|Cu|}) we get
	\[
	\sum_{\forall u\in \mathscr{A}^{(n)}}^{}9\left(\frac{\rho^{(n)}(u)}{d^2}\right) \geq n,
	\]
	thus,
	\[
	P^{(n)}_{tot} = \sum_{\forall u\in \mathscr{A}^{(n)}}^{} \rho^{(n)}(u) \geq \frac{nd^2}{9}.
	\]

\section{Proof of Proposition~\ref{prp:grd2}}
\label{app:grd2}

	For any given node $u \in \bar{\mathscr{A}^{(c)}}$, there are at least $\lceil\frac{1}{2}\lfloor\frac{m}{L}\rfloor\rceil$ transmitting rows either below or above that node. 
	Without loss of generality,  we suppose that node $u$ has at least $\lceil\frac{1}{2}\lfloor\frac{m}{L}\rfloor\rceil$ transmitting rows atop.
	Let $\mathcal{P}_r(u)$ be the sum of powers received at $u \in \bar{\mathscr{A}^{(c)}}$ from all nodes in all transmitting rows.  
	Among nodes in a non-transmitting row, the one on the rightmost column has the least value of $\mathcal{P}_r(u)$. 
	Hence, we assume that node $u$ is in rightmost column.
	Let $l$ denote the Euclidian distance between $u$ and the closest upper transmitting row.
	The total power received at $u$ from nodes in $i$th upper transmitting row, denoted as $\mathcal{P}_{i}(u)$, is
	\begin{equation}
	\begin{split}
	\label{equ:P_i}
	\mathcal{P}_{i}(u) 
	&= \sum_{j=0}^{m-1}\frac{d^2}{((l+(i-1)L)d)^2+(jd)^2} \\
	&\geq \sum_{j=0}^{m-1}\frac{d^2}{(iLd)^2+(jd)^2} \\
	&= \sum_{j=0}^{m-1}\frac{1}{(iL)^2+j^2}.
	\end{split}
	\end{equation}  
	Since the number of upper transmitting rows is at least $c =\left\lceil\frac{1}{2}\lfloor\frac{m}{L}\rfloor\right\rceil$, we get
	\begin{equation}
	\label{equ:P-r}
	\begin{split}
	\mathcal{P}_r(u) 
	\geq \sum_{i=1}^{c}\mathcal{P}_i(u)
	\geq \displaystyle \sum_{i=1}^{c} \sum_{j=0}^{m-1}\frac{1}{(iL)^2+j^2}, 
	\end{split}
	\end{equation}
	where the second inequality is by (\ref{equ:P_i}). 
	We have $m\geq iL$, \mbox{$1 \leq i \leq c$}. Thus,
	\begin{equation*}
	\label{equ:new-cond}
	\begin{split}
	\displaystyle \sum_{i=1}^{c} \sum_{j=0}^{m-1}\frac{1}{(iL)^2+j^2}
	\geq \sum_{i=1}^{c} \sum_{j=0}^{iL-1}\frac{1}{(iL)^2+j^2}
	\geq \sum_{i=1}^{c}\frac{1}{2Li}.
	\end{split}	
	\end{equation*}  
	Using the partial sum of harmonic series, we get
	\begin{equation}
	\begin{split}
	\displaystyle \sum_{i=1}^{c}\frac{1}{2Li} = \frac{1}{2L}\sum_{i=1}^{\left\lceil\frac{1}{2}\lfloor\frac{m}{L}\rfloor\right\rceil}\frac{1}{i} 
	&\geq  \frac{1}{2L}\ln\left\lceil\frac{1}{2}\lfloor\frac{m}{L}\rfloor\right\rceil \\
	&\geq  \frac{1}{2L}\ln\left(\frac{1}{2}(\frac{m}{L}-1)\right).\\
	\end{split}	
	\end{equation}
	Thus, $\mathcal{P}_r(u) \geq 1$ (hence, $u$ can decode the message) if
	\begin{equation}
	\label{equ:revised-cond}
	\frac{1}{2L}\ln((\frac{m}{2L}-\frac{1}{2})) \geq 1. 
	\end{equation}
	Finally, it can be verified that (\ref{equ:revised-cond}) holds for any \mbox{$L \leq 0.3\ln (m) = 0.15 \ln (n)$} and $n = m^2 \geq 4$, which completes the proof.

\end{appendices}








 

\bibliographystyle{IEEEtran}
\bibliography{Reference} 

\end{document}